\documentclass[letterpaper, DIV=13]{scrartcl}

\usepackage{amssymb, amsthm, mathtools,bbm}
\usepackage[square,sort,comma,numbers]{natbib}
\usepackage{enumitem}

\usepackage[none]{hyphenat}
\newtheorem{theorem}{Theorem}

\newtheorem{lemma}[theorem]{Lemma}

\newtheorem{proposition}[theorem]{Proposition}
\newtheorem{assumption}[theorem]{Assumption}
\numberwithin{equation}{section}
\numberwithin{theorem}{section}

\newcommand{\rr}{{\mathbb{R}}}
\newcommand{\cc}{{\mathbb{C}}}

\newcommand{\nn}{{\mathbb{N}}}

\newcommand{\tr}{{\operatorname{Tr}\,}}

\newcommand{\beq}[1]{\begin{equation} \label{#1}}
	\newcommand{\eeq}{\end{equation}}

\renewcommand{\epsilon}{\varepsilon}

\def\be{\begin{equation}}
	\def\ee{\end{equation}}

\begin{document}
	
\title{   Quantum Mean-Fields Spin Systems \\ in a Random External Field}

	\author{Chokri Manai\\
	{\normalsize Courant Institute of Mathematical Sciences, New York University}}
	\date{\vspace{-.3in}}
	\maketitle

	\minisec{Abstract}
	In this work, we consider general exchangeable quantum mean-field Hamiltonian such as the prominent quantum Curie-Weiss model under the influence of a random external field. Despite being arguably the simplest class of disordered quantum systems, the random external field breaks the symmetry of the mean-field Hamiltonian and hence standard quantum de Finetti type or semiclassical  arguments are not directly applicable. We introduce a novel strategy in this context, which can be seen as non-commutative large deviation analysis, allowing us to characterize the limiting free energy in terms of a simple and explicit variational formula. The proposed method is general enough to be used for other classes of mean-field models such as multi species Hamiltonians.
	
	
	\bigskip
	

\section{Introduction}\label{sec:intro}

Classical mean-field systems such as the classical Curie–Weiss model form an important class in statistical physics, as they are often exactly solvable and provide a good approximation to more realistic and complex lattice systems in high dimensions. Therefore, mean-field systems continue to play a central role in statistical physics, and many fundamental questions are canonically first addressed on this playground. In the case of disordered systems, in particular in spin glass theory, even today only mean-field systems such as the Sherrington–Kirkpatrick model are well understood \cite{SK75, Pa13, Ta06}, although significant recent progress has been made in proving glassy properties or the absence in the Edwards–Anderson model and in the random-field Ising model \cite{Ch15, Ch23, Ch24} The latter is a remarkably rich and challenging model reflected by the rounding phenomena in two dimensions \cite{AW89} and ferromagnetic behavior in higher dimensions \cite{BK88}; and only very recently the exponential decay of correlations has been established \cite{DX21}. Due to the complexity of the random-field Ising model, its corresponding mean-field variant is often studied as a simplification—namely, the Curie–Weiss model with a random external field \cite{BBI09, BdH15, FMP00, AF91, KL24}.

In the non-commutative case, the analysis of nearest-neighbor models is more often not tractable and, hence,  mean-field quantum spin systems are common in effective descriptions of a variety of phenomena. One example is the family of  Lipkin-Meshkov-Glick Hamiltonians, which were introduced in~\cite{LMG65-1,LMG65-2,LMG65-3}  to explain shape transitions in nuclei and are also used to describe quantum-spin tunneling in molecular magnets~\cite{MMBook14}. 
Another instance is the 
quantum Curie-Weiss Hamiltonian, which continues to draw the attention of many communities ranging from statistical physics to quantum computing~\cite{Bapst_2012,Chayes:2008aa,Landsman:2020aa,Bjornberg:2020aa,Bulekov:2021aa, MW23}. Due to the rapid progress and sustained interest in the understanding of quantum disordered models such as random matrix ensembles \cite{DG09, EPRS10, PS11, TV11}, the Anderson model \cite{AW13, AW15, DS20} and quantum glasses~\cite{KMW25, LMRW21, LRRS21, MW20, MW21, MW21b, MW22, MW23b, MW25}, it appears canonical to investigate the behavior of quantum spin systems under the influence of random external fields. And here the most natural starting point is the analysis of the thermodynamics of quantum mean-field models in the presence of a random field. Surprisingly, this simple class of disordered models has not been  analyzed yet and we aim to address this gap in the mathematical (and physical) literature. At first glance, one might expect that the methods used to solve deterministic mean-field quantum models carry over with some technical modifications to random field models.  However, the two main approaches in the literature, Lieb-Berezin bounds \cite{Ber72, Lieb73} and quantum de-Finetti arguments \cite{FSV80}, are very rigid and are largely based on the full exchangeability of the Hamiltonian and, hence, there is no easy adaptation to the apparently slight perturbation by a random field. Our second main purpose in this work is to introduce a strategy which is flexible enough to solve more general mean-field models. Our approach consists of two main ideas. First, we control the thermodynamical fluctuations of the magnetization by  a perturbation argument which originated in spin glass theory but seems to be not so well-known in other fields \cite{STT20, ST93}. This strategy is even useful for purely deterministic Hamiltonians.  Our second ingredient  is a novel approach for a microcanonical analysis in the non-commutative setting and allows us to essentially derive a non-commutative version of Varadhan's lemma in the context of mean-field models. The main new observation is that while in the quantum case the analysis of linear models is not enough to derive the free energy of the system restricted to a fixed magnetization, one only needs to further understand quadratic models to carry out a microcanonical analysis. We stress that our proof strategy is  not restricted to random-field models and we expect that it can be used to derive further novel free energy formulae - for example in the interesting case of multi-species Hamiltonians.

Let us give a precise description of the models of interest. We consider $ N $ interacting qubits, and  the corresponding Hilbert space  is the tensor product $ \mathcal{H}_N = \bigotimes_{n=1}^N \mathbb{C}^2 $. On these finite-dimensional Hilbert spaces we consider a Hamiltonian or more precisely a family of Hamiltonians
\begin{equation}
	H_{V, \mathfrak{b}} := H_{V} + H_{\mathfrak{b}},
\end{equation}
where $H_{V}$ models a deterministic mean-field spin system and $H_{\mathfrak{b}}$ the contribution of the random external field.  The parameters $V$ corresponds to a "semiclassical symbol" which is a continuous function $V : B_{\rr^3} \to \rr$ on the 3-dimensional ball  $B_{\rr^3}$ and $\mathfrak{b} = (\mathfrak{b}_x, \mathfrak{b}_y, \mathfrak{b}_z)$ is a vector-valued integrable random variable. 
  For the most prominent mean-field spin Hamiltonians, the symbol $V$ is a polynomial and in this case the mean-field Hamiltonian allows for a simple description in terms of spin operators. We focus on this situation in the rest of the introduction and postpone the slightly more general consideration to Section~\ref{sec:model}.    
 
For $1 \leq n \leq N$, the matrix-valued vector $ \mathbf{S}(n) = \mathbbm{1} \otimes \dots \otimes \ \mathbf{s} \ \otimes \dots \otimes \mathbbm{1}  $ stands for the natural lift of 
the spin vectors $  \mathbf{s}  = (s_x, s_y, s_z ) $ to the $ n $-th component of the tensor product. On each copy of $ \mathbb{C}^2 $, the spin vector $\mathbf{s}$ coincides with the three generators  of $ SU(2) $:
\begin{equation*}
s_x = \frac{1}{2} \left( \begin{matrix} 0 & 1 \\ 1 & 0 \end{matrix}\right) , \quad s_y = \frac{1}{2} \left( \begin{matrix} 0 & -i \\ i & 0 \end{matrix}\right) , \quad s_z = \frac{1}{2} \left( \begin{matrix} 1 & 0 \\ 0 & -1 \end{matrix}\right) . 
\end{equation*}
Then, the  total spin-vector is defined as
$$	 \mathbf{S} = \sum_{n=1}^N  \mathbf{S}(n). $$
If $V = \textrm{P}$ is a polynomial, the  mean-field Hamiltonian takes the form
\begin{equation}\label{eq:Ham}
	H_P = N \ \textrm{P}\Big(\tfrac{2}{N}   \mathbf{S} \Big),
\end{equation}
 where the non-commuting polynomial $  \textrm{P}\Big(\tfrac{2}{N}   \mathbf{S} \Big) $ is understood as the Weyl-ordered representant of $P$. More precisely,  $  \textrm{P}\Big(\tfrac{2}{N}   \mathbf{S} \Big) $ is a finite linear combination of products of the  three rescaled components of the total spin-vector $\mathbf{S}$ in which each product is averaged under the three components so that  it becomes a self-adjoint operator. For instance, if $ P( \mathbf{m}) = m_x m_y $ we obtain  $ \textrm{P}\big(2   \mathbf{S} /N \big) =2 \big(  \mathbf{S}_x  \mathbf{S}_y +  \mathbf{S}_y  \mathbf{S}_x \big) / N^2.$   The Weyl-ordering guarantees that the associated mean-field Hamiltonian
is always self-adjoint on $ \mathcal{H}_N $. The dependence on the particle number is two-fold and quintessential for the mean-field nature. 
The scaling of the spin ensures that the operator is norm-bounded by one, $ \left\| \tfrac{2}{N}  S_\xi \right\| \leq 1 $ for all $ \xi= x,y, z $. Moreover, the prefactor $ N $ forces the energy $ H $  to be extensive.  

Let us illustrate the richness of this class of Hamiltonians by discussion a few examples. The simplest case constitute the polynomials $P( \mathbf{m}) = m_z^p$ which correspond to the classical and commutative p-spin Curie-Weiss models and $p=2$ is of course the standard Curie-Weiss (CW) model. A very popular non-commutative model is the Quantum Curie-Weiss (QCW) model which corresponds to the choice $P( \mathbf{m}) = m_z^2 + \gamma m_x$, where $\gamma \geq 0$ can be interpreted as the strength of an external transverse field. A particularly interesting feature of the QCW is that it allows for a probabilistic Poisson path integral representation which links its thermodynamics to percolation type problems. However, this beautiful approach yields rather weak results compared to a spectral analysis of the Hamiltonian \cite{Bjornberg:2020aa, Chayes:2008aa}. More generally, the well-studied Lipkin-Meshkov-Glick  model is given by
$ P( \mathbf{m}) = -\alpha m_y^2 - \beta m_z^2 -\gamma m_x $ with $ \alpha,\beta , \gamma \in \mathbb{R} $.

The definition of the random external field is much less involved. Indeed, given an integrable vector valued random variable $\mathfrak{b}$, we consider an i.i.d. copies $\mathbf{b}(1), \ldots, \mathbf{b}(N)$ of $\mathfrak{b}$ and define
\begin{equation}\label{eq:rf}
	H_{\mathfrak{b}} := 2 \sum_{n = 1}^{N} \langle \mathbf{b}(n),  \mathbf{S}(n) \rangle := 2 \sum_{n = 1}^{N} b_{x}(n) S_{x}(n) + b_{y}(n) S_{y}(n) + b_{z}(n) S_{z}(n), 
\end{equation}
where the factor 2 is introduced to simplify a few expressions in the forthcoming discussion. Of course, more precisely the Hamiltonian(s) 	$H_{\mathfrak{b}}$, and thus also the total Hamiltonians $H_{V, \mathfrak{b}}$, form a family of matrix-valued random variables $H_{\mathfrak{b}}(\omega)$ on some common probability space $(\Omega, \mathcal{F}, \mathbb{P})$. As the underlying probability space will not be of particular importance, this detail will be ignored for the most part in the following.  

Having introduced the random field mean-field Hamiltonians, we recall some basic quantities of statistical physics. The (random) partition function is given by 
\begin{equation}\label{eq:Z}
	Z_N(V,\mathfrak{b}) := \tr e^{H_{V, \mathfrak{b}}}
\end{equation}
and the specific pressure or free energy is
\begin{equation}
	p_N(V,\mathfrak{b}) := \frac1N \log Z_N(V,\mathfrak{b}).
\end{equation}
The canonical definition of the partition further contains  the inverse temperature $\beta \geq 0$ in the exponent and the dependence of the pressure on $\beta$  is crucial for an understanding of phase transitions in a given model. Since in our case $\beta H_{V, \mathfrak{b}} = H_{\beta V, \beta \mathfrak{b}}$, we choose to include $\beta$ in the parameters of the model.

Our main result is a variational characterization of the limit of $\mathbb{E}[p_N(V,\mathfrak{b})]$ where $\mathbb{E}$ denotes here and in the following the disorder average, i.e., the average with respect to the random variables $b_1, \ldots, b_N$. To formulate our theorem, we need to introduce the two real valued functions $\Lambda_\mathfrak{b}, \Lambda_\mathfrak{b}^{*} : \rr^3 \to \rr$, where $\Lambda_\mathfrak{b}$ is defined as
\begin{equation}\label{eq:Lambda}
	\Lambda_\mathfrak{b}(\mathbf{h}) = \mathbb{E}\left[\log 2 \cosh \left(  \sqrt{(h_x+\mathfrak{b}_x)^2 + (h_y+\mathfrak{b}_y)^2 + (h_z+\mathfrak{b}_z)^2}   \right) \right]
\end{equation}
and $\Lambda_\mathfrak{b}^{*}$ is its Legendre transform
\begin{equation}\label{eq:Lambda*}
	\Lambda^{*}_\mathfrak{b}(\mathbf{m})= \sup_{\mathbf{h} \in \rr^3} \bigg( \langle \mathbf{m}, \mathbf{h} \rangle - \Lambda_\mathfrak{b}(\mathbf{h}) \bigg),
\end{equation}
where $\langle \cdot, \cdot \rangle$ denotes the standard Euclidean scalar product. The reader may have realized that $\Lambda_\mathfrak{b}(\mathbf{h})$ coincides with the $N$-independent expected value of the specific  pressure of a linear field Hamiltonian $ H = \sum_{n = 1}^{N} 2(b_x(n) + h_x) S_{x}(n) + 2(b_y(n) + h_y) S_y(n) + 2(b_z(n) + h_z) S_z(n)$. The negative of the Legendre transform  will be identified later as restricted free energy and is the non-commutative substitute of classical large deviation rate functions. Our main theorem reads as follows.

\begin{theorem}\label{thm:main}
Let $V : B_{\rr^3} \to \rr$ be a continuous function on the three-dimensional unit ball, $\mathfrak{b}$ an integrable (i.e. an $L^1$-)vector-valued random variable and $ H_{V, \mathfrak{b}}$ the corresponding family of self-adjoint random Hamiltonians defined on the $N$-particle Hilbert spaces $ \mathcal{H}_N = \bigotimes_{n=1}^N \mathbb{C}^2. $ Then, the quenched expectations $\mathbb{E}[p_N(V,\mathfrak{b})]$ converge and the limit is given by
\begin{equation}\label{eq:main}
	\lim_{N \to \infty} \mathbb{E}[p_N(V,\mathfrak{b})] = p(V,\mathfrak{b}) := \sup_{\mathbf{m} \in B_{\rr^3}} \bigg(V(\mathbf{m}) - \Lambda^{*}_\mathfrak{b}(\mathbf{m}) \bigg).
\end{equation}
In fact, if the random variables $p_N(V,\mathfrak{b})$ are defined on a common probability space, then they converge almost surely and in $L^1$ to the deterministic value $p(V,\mathfrak{b})$.
\end{theorem}

The proof of Theorem~\ref{thm:main} is spelled out in Section~\ref{sec:proof}. Even though the definition of $H_V$ for general continuous functions is most conveniently formulated in terms of Bloch coherent states (see Section~\ref{sec:model}), the proof for polynomial models is self-contained and does not rely on any tools from semiclassical analysis. We close the introduction with a few remarks on Theorem~\ref{thm:main}.
\begin{enumerate}
	\item We first note that the integrability of $\mathfrak{b}$ ensures that $	\Lambda_\mathfrak{b}(\mathbf{h})$ exists for every $\mathbf{h} \in \mathbb{R}^3$ and we have the bounds $ | \mathbb{E} \mathfrak{b} + \mathbf{h}| \leq \Lambda_\mathfrak{b}(\mathbf{h}) \leq \Lambda_\mathfrak{b}(\mathbf{0}) + |\mathbf{h}|$, which will be proved below as part of Lemma~\ref{lem:lambda}. These inequalities ensure that the Legendre transform $	\Lambda^{*}_\mathfrak{b}(\mathbf{m}) < \infty$ for all $\mathbf{m} \in B_{\rr^3}$. As a convex function, $\Lambda^{*}_\mathfrak{b}(\mathbf{m})$ is continuous and hence there exists a maximizer $\mathbf{m}$ for the free energy \eqref{eq:main}.
	\item In the special case $\mathfrak{b} = 0$ almost surely, the Hamiltonian reduces to the deterministic mean-field part. We demonstrate now that \eqref{eq:main} is consistent with the known expression in this case. First, $\Lambda_0(\mathbf{h}) = \log 2\cosh(|\mathbf{h}|)$ is radial. So its Legendre transform is also radial and an elementary computation shows
	\begin{equation}
		\Lambda_0^{*}(\mathbf{m}) =  \log \frac{1+|\mathbf{m}|}{2} + \log \frac{1-|\mathbf{m}|}{2} =: - I(|\mathbf{m}|)
	\end{equation}
	with the binary entropy $I \colon [0,1] \to \rr$. Hence, we obtain 
	\begin{equation}\label{eq:maindetermin}
		p_{V,0} := \sup_{\mathbf{m} \in B_{\rr^3}} \bigg( V(\mathbf{m} + I(|\mathbf{m}|) \bigg) = \max_{r \in [0,1]} \left\{ I(r) +  \max_{\Omega \in S^2}   V(r\mathbf{e}(\Omega) )  \right\} 
	\end{equation}
	where $\Omega =(\theta, \varphi) $ is the spherical parametrization of the  unit sphere $ S^2 $ with $ 0\leq \theta \leq \pi $,\; $ 0 \leq \varphi \leq 2 \pi $ and $\mathbf{e}(\Omega)$ is the corresponding unit vector. This is a known variational formula for $p_{V,0}$ (see e.g. \cite{MW23}).
	
	\item The proper setting for almost sure and $L^1$-convergence is a follows. Let $\mathbf{b}(1), \mathbf{b}(2), \ldots$ be an infinite sequence of i.i.d. random variables on a probability space $(\Omega, \mathcal{F}, \mathbb{P})$ where each $\mathbf{b}(i)$ shares the distribution of $\mathfrak{b}$. Now, $(H_{V,\mathfrak{b}}(\omega))$ are matrix-valued random variables and $p_N(V,\mathfrak{b})$ are real-valued random variables on  $(\Omega, \mathcal{F}, \mathbb{P})$. Note that the random fields $\mathbf{b}(i)$ are reused for different particle numbers $N$ and hence the pressures $p_N(V,\mathfrak{b})$ are not independent from each other (however one would still obtain the same results if we considered independent $p_N(V,\mathfrak{b})$). Theorem~\ref{thm:main} says that these $p_N(V,\mathfrak{b})$ converge almost surely and in $L^1$ to the constant $p(V,\mathfrak{b})$. Note that the almost sure convergence is not trivial as the pressure may have an unbounded variance.
	
	\item In disordered systems, the annealed free energy,
	\begin{equation}
		p_N^{\mathrm{ann}}(V,\mathfrak{b}) := \frac1N \log \left( \mathbb{E}[ Z_N(V,\mathfrak{b})] \right),
	\end{equation}
	is often considered as simplification for the more physical quenched free energy. In random field models, annealed and quenched averages disagree (even for very high temperatures). If $\mathbb{E}[\exp(|\mathfrak{b}_{x,y,z}|)] < \infty$ (otherwise $p_N^{\mathrm{ann}}(V,\mathfrak{b}) = \infty$ for all $N$ ), we have a similar variational formula 
	\begin{equation}
		\lim_{N \to \infty} p_N^{\mathrm{ann}}(V,\mathfrak{b}) = \sup_{\mathbf{m} \in B_{\rr^3}} \left(V(\mathbf{m}) - (\Lambda^{\mathrm{ann}}_\mathfrak{b})^{*}(\mathbf{m}) \right),
	\end{equation}
	where $(\Lambda^{\mathrm{ann}}_\mathfrak{b})^{*}$ is the Legendre transform of the annealed constant field pressure
	\begin{equation}
		\Lambda^{\mathrm{ann}}_\mathfrak{b}(\mathbf{h}) = \log \left( 2 \,  \mathbb{E} \left[ \cosh \left( \sqrt{(h_x+\mathfrak{b}_x)^2 + (h_y+\mathfrak{b}_y)^2 + (h_z+\mathfrak{b}_z)^2}  \right)  \right] \right)
	\end{equation}
\end{enumerate}

\section{Coherent States and the Hamiltonian for general $V$}\label{sec:model}

To give a precise definition of $H_V$ for general continuous functions $V$, we need to recall some further facts about polynomial Hamiltonians $H_P$ and some tools from semiclassical analysis. We first note that $ H_P $ in~\eqref{eq:Ham} is a function of the total spin $  \mathbf{S}. $ Thus, $ H_P $ is block diagonal with respect to the decomposition of the  tensor-product Hilbert space $\mathcal{H}_N$ according to the irreducible representations of the total spin corresponding to the eigenspaces of  $ \mathbf{S}^2 $ with eigenvalues $ J(J+1) $, i.e.
\begin{equation}\label{def:blockH}
\mathcal{H}_N \equiv  \bigoplus_{J=\frac{N}{2} - \lfloor \frac{N}{2}\rfloor}^{N/2} \bigoplus_{\alpha=1}^{M_{N,J}}  \; \mathbb{C}^{2J+1} , \qquad M_{N,J}  = \frac{2J+1}{N+1} \binom{N+1}{\frac{N}{2} +J + 1 } . 
\end{equation}
The total spin $ J $ of $ N $ qubits can take any value from $ N/2 $ down in integers to either $ 1/2 $ if $ N $ is odd or $ 0 $ if $ N $ is even. The degeneracy of the representation of spin $ J $ in this decomposition is $ M_{N,J}  $ \cite{Mihailov:1977aa}. 
On each block $ (J, \alpha) $, the Hamiltonian~\eqref{eq:Ham} then acts as the given polynomial of the generators of the irreducible representation of $ SU(2) $ on $ \mathbb{C}^{2J+1} $.

It is well known that for large spin quantum numbers $J$ the spin operators behave in leading order as classical vector and the rigorous formulation of semiclassical analysis is most conveniently carried out via Bloch coherent states on the Hilbert space $  \mathbb{C}^{2J+1}$ \cite{Lieb73,MN04,Biskup:2007aa}.  They are parametrized by an angle $\Omega =(\theta, \varphi) $ on the unit sphere $ S^2 $ with $ 0\leq \theta \leq \pi $,\; $ 0 \leq \varphi \leq 2 \pi $. In bra-ket-notation, which we will use in this paper, the Bloch-coherent states are given by
\begin{equation}\label{def:cs}
\big| \Omega, J \rangle := U(\theta,\varphi) \  \big| J \rangle , \qquad U(\varphi,\theta) :=  \exp\left( \frac{\theta}{2} \left( e^{i\varphi} S_- - e^{-i\varphi} S_+ \right) \right)  
\end{equation} 
The reference vector  $  \big| J \rangle \in  \mathbb{C}^{2J+1}  $ is the normalized eigenstate of the operator $S_z$ corresponding to the (maximal) eigenvalue $ J $ on the Hilbert space $ \mathbb{C}^{2J+1} $. The operators $ S_\pm = S_x \pm i S_y $ are the spin raising and lowering operators of the irreducible representation of $ SU(2) $ on $  \mathbb{C}^{2J+1} $. 

We recall some properties of Bloch coherent states. First and foremost, they form an overcomplete set of vectors as expressed through the resolution of unity on  $  \mathbb{C}^{2J+1} $:
\begin{equation}\label{eq:completeness}
 \frac{2J+1}{4\pi} \int \big| \Omega, J \rangle \langle \Omega, J \big|  \ d\Omega  = \mathbbm{1}_{  \mathbb{C}^{2J+1}} .
\end{equation}
Every linear operator $  G $ on $  \mathbb{C}^{2J+1} $ is associated with a lower and upper symbol. The lower symbol is  $ G(\Omega,J) :=  \langle \Omega , J \big|  G  \big| \Omega, J \rangle $, and the upper symbol is characterized  through the property
\begin{equation}
 G = \frac{2J+1}{4\pi} \int  g(\Omega,J) \,  \big| \Omega, J \rangle \langle \Omega, J \big| \, d\Omega  .
 \end{equation}
The choice of $ g $ is not unique. E.g.\ through an explicit expression~\cite{Kutzner:1973rz}, one sees that there is always an arbitrarily often differentiable choice,  $ g(\cdot, J) \in C^\infty(S^2)  $.  
More properties of coherent states can be found in~\cite{Per86,CR12,Zwo12}.
The lower and upper symbol feature prominently in Berezin and Lieb's semiclassical bounds~\cite{Ber72,Lieb73,Simon:1980} on the partition function associated with a self-adjoint Hamiltonian $ G $ on $ \mathbb{C}^{2J+1} $:
\begin{equation}\label{eq:BerezinLieb}
 \frac{2J+1}{4\pi} \int e^{ G(\Omega,J)} d\Omega \leq  \tr_{\mathbb{C}^{2J+1}}  e^{G } \leq \frac{2J+1}{4\pi} \int e^{g(\Omega,J)} d\Omega  . 
 \end{equation}
 For many physical models,  the lower and upper symbol asymptotically coincide for large spin quantum numbers $ J $ and, hence, the bound \eqref{eq:BerezinLieb} becomes sharp \cite{Lieb73,Simon:1980,Duffield:1990aa}. 
Indeed, for any polynomial of the spin operator as in~\eqref{eq:Ham} restricted to $  \mathbb{C}^{2J+1} $ both the upper and lower symbols agree to leading order in $ N $ with the corresponding classical polynomial function on the unit ball $ B_{\rr^3}. $ In spherical coordinates $ \mathbf{e}(\Omega) = \left(\sin\theta \cos\varphi ,  \sin\theta \sin\varphi , \cos\theta \right) \in S^2 $, we have 
\begin{equation}\label{eq:uppersymbclasss}
 \sup_{0 \leq J \leq N/2} \left\|  \textrm{P}\Big(\tfrac{2}{N}  \mathbf{S} \Big)\Big|_{  \mathbb{C}^{2J+1}} -  \frac{2J+1}{4\pi} \int    \textrm{P}\Big(\tfrac{2J}{N} \mathbf{e}(\Omega)  \Big)  \big| \Omega, J \rangle \langle \Omega, J \big| \, d\Omega  \right\| \leq \mathcal{O}(N^{-1}) 
\end{equation}
for the operator norm $ \| \cdot \| $ on $  \mathbb{C}^{2J+1} $. Here, we used the Landau $ \mathcal{O} $-notation, i.e., the error on the right is bounded by $ C N^{-1} $ with a constant $ C $ which only depends on the coefficients of the polynomial. This statement is a quantitative version of Duffield's theorem~\cite{Duffield:1990aa} and was first established in this form in \cite{MW23}.  There it has also been shown that an upper bound symbol as \eqref{eq:uppersymbclasss} also implies 
that the lower symbol shares the same classical asymptotics, i.e. for the polynomial case
$$  \sup_{0 \leq J \leq N/2} \sup_{\Omega} \left| \langle \Omega , J \big|   \textrm{P}\big(\tfrac{2}{N}  \mathbf{S} \big)\Big|_{  \mathbb{C}^{2J+1}}  \big| \Omega, J \rangle -    \textrm{P}\Big(\tfrac{2J}{N} \mathbf{e}(\Omega)  \Big)  \right| \leq  \mathcal{O}(N^{-1})  .
$$
In view of the Berezin-Lieb inequalities~\eqref{eq:BerezinLieb}, the Duffield-type bound \eqref{eq:uppersymbclasss} captures  the thermodynamic behavior of the model and this observation motivates the following assumption for more general mean-field type Hamiltonians $H_V$.  

\begin{assumption}
	Let $V : B_{\rr^3} \to \rr$ be a continuous function. We assume that 
	$ H_V $ is block diagonal with respect to the orthogonal decomposition \eqref{def:blockH} of $ \mathcal{H}_N $,
	\begin{equation}\label{eq:Hblock}
		H =  \bigoplus_{J=\frac{N}{2} - \lfloor \frac{N}{2}\rfloor}^{N/2} \bigoplus_{\alpha=1}^{M_J}  H_{J,\alpha} 
	\end{equation}
	with self-adjoint blocks $  H_{J,\alpha}  $ acting on a copy of $ \mathbbm{C}^{2J+1} $. 
	Moreover, 
	all block Hamiltonians  are uniformly approximable by $V$ as an upper symbol in operator norm on $ \mathbb{C}^{2J+1} $ to order $N$ as $ N \to \infty $:
	\begin{equation}\label{ass:symbol}
		\max_{J,\alpha}  \left\| H_{J,\alpha}  -    \frac{2J+1}{4\pi} \int  N V\Big(\frac{2J}{N}  \mathbf{e}(\Omega) \Big) \,  \big| \Omega, J \rangle \langle \Omega, J \big| \, d\Omega  \right\| = o(N) ,
	\end{equation}
	where the maximum runs over $ \alpha  \in \{ 1, \dots , M_{N,J} \}  $ and $ J \in \{ \frac{N}{2} - \lfloor \frac{N}{2}\rfloor, \dots , N/2\} $.
\end{assumption}

Note that strictly speaking, the symbol $V$ does not uniquely characterize the deterministic sequence of operators $H_V$, but rather determines a class of operators whose asymptotic behavior to be approximable by $V$ well enough.  We further note that the degeneracy numbers $M_{N,J}$ agree on exponential order with the binary entropy $I(r)$ evaluated at $r = N/(2J)$ and in combination with the Berezin-Lieb inequalities these readily yields (see also \cite{MW23})

\begin{equation*}
p_V := \lim_{N\to \infty} N^{-1} \ln \tr \exp\left( H_V \right) = \max_{r \in [0,1]} \left\{ I(r)  \max_{\Omega \in S^2}   V\left(r\mathbf{e}(\Omega) \right)  \right\},
\end{equation*}
which is of course exactly \eqref{eq:maindetermin}.  We emphasize that this elegant approach cannot be extended to the disordered model as it crucially depends on the block diagonal structure and the approximability by single coherent states and both properties are not shared by the random field model.

\section{Proof of Theorem~\ref{thm:main} via a Non-Commutative Microcanonical Approach}\label{sec:proof}

We will prove Theorem~\ref{thm:main} for polynomial models by establishing an asymptotically sharp lower bound in Section~\ref{sec:lower} and the corresponding upper bound in Section~\ref{sec:upper}. The latter requires some preparations, most notably an apriori control of thermal fluctuations which is the content of Section~\ref{sec:therm}. The extension to more general continuous symbols $V$ is carried out in Section~\ref{sec:con}.

\subsection{Sharp Lower Bound by Gibbs Variational Principle}\label{sec:lower}

As usual, obtaining sharp lower bounds for the pressure is the easier part, as by Gibbs variational principle we only need to find a good trial state. Indeed, recall that Gibbs' principle asserts that for any self-adjoint operator $H$ on a finite dimensional Hilbert space $\mathcal{H}$ we have
\begin{equation}\label{eq:Gibbs}
	\log \tr e^{H} = \sup_{\varrho \in \mathcal{D}_{\mathcal{H}}} \left( \tr[\varrho H] - \tr[ \varrho \log \varrho] \right),
\end{equation}
where $\mathcal{D}_{\mathcal{H}}$ denotes the set of density matrices
	\begin{equation}
		\mathcal{D}_{\mathcal{H}} := \{ \varrho \in \mathcal{B}{\mathcal{H}} \, | \, \varrho \geq 0, \, \tr \varrho = 1 \}
	\end{equation}
and $\log \varrho$ is defined via the spectral calculus for a positive semidefinite $\rho$. We pick (random) trial states $\varrho_{\mathfrak{b}}(\mathbf{h})$ on $\mathcal{H}_N$ of the form 
\begin{equation}\label{eq:trial}
	\varrho_{\mathfrak{b}}(\mathbf{h}) := \frac{e^{H_{\mathfrak{b}} + 2 \langle \mathbf{h}, \mathbf{S} \rangle }}{\tr e^{H_{\mathfrak{b}} + 2 \langle \mathbf{h}, \mathbf{S} \rangle }} = \bigotimes_{n=1}^{N} \frac{e^{ 2 \langle \mathbf{h} + \mathbf{b}(n) , \mathbf{S}(n) \rangle }}{\tr e^{ 2 \langle \mathbf{h} + \mathbf{b}(n) , \mathbf{S}(n) \rangle }} := \bigotimes_{n=1}^{N} \varrho_{\mathfrak{b}}^{(n)}(\mathbf{h}),
\end{equation}
where the product identity follows from the tensor product structure of linear fields. Using the states from \eqref{eq:trial}, we want to establish the following
\begin{proposition}\label{prop:lower}
	Let $\textrm{P} \colon \rr^3 \to \rr $ be a polynomial, $\mathfrak{b}$ an integrable vector valued random variable and $H_{\textrm{P},\mathfrak{b}}$ the associated random-field Hamiltonian(s) with specific pressures $p_N(\textrm{P},\mathfrak{b})$. Then, almost surely
	\begin{equation}\label{eq:lower}
	  \sup_{\mathbf{m} \in B_{\rr^3}} \left(\textrm{P}(\mathbf{m}) - \Lambda^{*}_\mathfrak{b}(\mathbf{m}) \right) \leq 	\liminf_{N \to \infty} p_N(\textrm{P},\mathfrak{b})
	\end{equation}
\end{proposition}

We have already sketched the general proof strategy, but we need some further technical considerations. We start by collecting some properties of the functional $\Lambda_\mathfrak{b}, \Lambda^{*}_\mathfrak{b}$.
\begin{lemma}\label{lem:lambda}
	Let $\mathfrak{b}$ be an integrable vector-valued random variable.
	\begin{enumerate}
		\item The function $\Lambda_{\mathfrak{b}}(\mathbf{h})$ is 1-Lipschitz, i.e. $| \Lambda_{\mathfrak{b}}(\mathbf{h}) - \Lambda_{\mathfrak{b}}(\mathbf{h}') | \leq |\mathbf{h} - \mathbf{h}'|$ and
		\begin{equation}\label{eq:lambdaineq}
			| \mathbb{E} [\mathfrak{b}] + \mathbf{h}| \leq \Lambda_\mathfrak{b}(\mathbf{h}) \leq \Lambda_\mathfrak{b}(\mathbf{0}) + |\mathbf{h}|
		\end{equation}
		\item The function $\Lambda_{\mathfrak{b}}(\mathbf{h})$ is strictly convex and for each $\mathbf{m} \in \mathring{B}_{\rr^3}$ there exists a unique $\mathbf{h} = \mathbf{h}(\mathbf{m})\in \rr^3$ such that $\mathbf{m} = \nabla \Lambda_{\mathfrak{b}}(\mathbf{h})$. For this vector $\mathbf{h}(\mathbf{m})$, we have $$\Lambda_{\mathfrak{b}}^{*}(\mathbf{m}) = \langle m, \mathbf{h}(\mathbf{m}) \rangle - \Lambda_{\mathfrak{b}}(\mathbf{h}(\mathbf{m})). $$
		\item The Legendre transform $\Lambda_{\mathfrak{b}}^{*}$ is bounded, continuous and convex on the unit ball $B_{\rr^3}$. 
		
	\end{enumerate}
\end{lemma}

\begin{proof}
	\begin{enumerate}
		\item For the first part, it is convenient to recall that $\Lambda_\mathfrak{b}$ is given by 
		$$ \Lambda_\mathfrak{b} = \mathbb{E} \log \tr e^{2 \langle \mathfrak{b} + \mathbf{h}, \mathbf{s} \rangle } $$
		with the spin vector $\mathbf{s}$ on $\cc^2$. For example, since the matrix norm $ \| 2 \langle  \mathbf{h}, \mathbf{s} \rangle \| = |h|$ the simple operator inequality $\langle \mathfrak{b} + \mathbf{h}, \mathbf{s} \rangle \leq \langle \mathfrak{b} , \mathbf{s} \rangle + |h| \mathbbm{1}$ yields the upper bound in \eqref{eq:lambdaineq}. A similar argument yields the Lipschitz continuity. For the lower bound, we recall that the map $A \mapsto \log \tr e^{A}$ is convex on the space of symmetric matrices. Hence, by Jensen's inequality
		$$  \Lambda_\mathfrak{b} = \mathbb{E} \log \tr e^{2 \langle \mathfrak{b} + \mathbf{h}, \mathbf{s} \rangle } \geq \log \tr e^{2 \langle \mathbb{E}[\mathfrak{b}] + \mathbf{h}, \mathbf{s} \rangle } \geq \|2 \langle \mathbb{E}[\mathfrak{b}] + \mathbf{h}, \mathbf{s} \rangle \| = | \mathbb{E} [\mathfrak{b}] + \mathbf{h}|.$$
		\item To prove the strict convexity of $\Lambda_{\mathfrak{b}}(\mathbf{h})$, it is enough to show that $\mathbf{h} \mapsto \log \tr e^{2 \langle \mathbf{b} + \mathbf{h}, \mathbf{s} \rangle }$ is strictly convex for any fixed $\mathbf{b} \in \rr^3$. Since a fixed $\mathbf{b}$ only translates the function, it is in fact enough to show that $\mathbf{h} \mapsto \log \tr e^{2 \langle  \mathbf{h}, \mathbf{s} \rangle }$ is strictly convex. Due to radial symmetry, it remains to show that the real function $f(r) := \log \cosh r $ is strictly convex. But $f''(r) = (\tanh(r))' = \frac1{\cosh^2(r)} > 0$ for all $r \in \rr$, which establishes the strict convexity. We fix now $\mathbf{m} \in \mathring{B}_{\rr^3}$ and consider the function $L(\mathbf{h}) := \langle \mathbf{m}, \mathbf{h} \rangle - \Lambda_{\mathfrak{b}}(\mathbf{h})$. Due to the lower bound in \eqref{eq:lower}, this function tends to $-\infty$. Hence, $L(\mathbf{h})$ attains its maximum at some $\mathbf{h}$ and since it is a critical point we have $\mathbf{m} = \nabla \Lambda_{\mathfrak{b}}(\mathbf{h})$. Since $L(\mathbf{h})$ is strictly concave, it follows that there is at most one critical point and thus $\mathbf{h}$ is the unique solution. The claim on the Legendre transform is immediate.
		
		\item It suffices to show boundedness since  Legendre transforms are always convex and every convex and bounded function is also continuous. However, for $\mathbf{m} \in B_{\rr^3}$ 
		$$ \langle \mathbf{m}, \mathbf{h} \rangle - \Lambda_\mathfrak{b}(\mathbf{h}) \leq |\mathbf{h}| - | \mathbb{E} [\mathfrak{b}] + \mathbf{h}| \leq |\mathbb{E} [\mathfrak{b}] | $$
		and hence $\Lambda_{\mathfrak{b}}^{*}(\mathbf{m}) \leq |\mathbb{E} [\mathfrak{b}] |.$ 
		
	\end{enumerate}
\end{proof}
	To shorten our notation, it is convenient to introduce the magnetization operator $$\mathbf{M} := \frac{2}{N} \mathbf{S},$$
	 which will appear frequently in the following. The key to the proof of Proposition~\ref{prop:lower} is the following
	 \begin{lemma}\label{lem:lower}
	 	Let $\mathbf{w} \in \rr^3$ be a unit vector and $d \in \nn$. Then, we have almost surely 
	 	\begin{equation}
	 		\lim_{N \to \infty} \tr \left[ \varrho_{\mathfrak{b}}(\mathbf{h}) (\langle \mathbf{w}, \mathbf{M} \rangle )^d \right] =  (\langle \mathbf{w}, \mathbf{m} \rangle )^d,
	 	\end{equation}
	 	where $\mathbf{m} = \nabla \Lambda_\mathfrak{b}(\mathbf{h})$ and $\varrho_{\mathfrak{b}}(\mathbf{h})$ are the states from \eqref{eq:trial}.
	 \end{lemma}
\begin{proof}
We start with the special case $d = 1$. Then, the operator $\langle \mathbf{w}, \mathbf{M} \rangle$ is also of tensor product form and we have 
$$ \tr \left[ \varrho_{\mathfrak{b}}(\mathbf{h}) \langle \mathbf{w}, \mathbf{M} \rangle  \right] =  \frac{2}{N} \sum_{n = 1}^{N} \left[ \varrho_{\mathfrak{b}}^{(n)}(\mathbf{h}) \langle \mathbf{w}, \mathbf{S}(n) \rangle  \right] $$
and the right-hand side is an average of bounded i.i.d. random variables. Hence by the strong law of large numbers, we have almost surely
\begin{align*} \lim_{N \to \infty} \tr \left[ \varrho_{\mathfrak{b}}(\mathbf{h})\langle \mathbf{w}, \mathbf{M} \rangle \right] &= \mathbb{E} \left[ \tr \varrho_{\mathfrak{b}}^{(1)}(\mathbf{h}) \langle \mathbf{w}, 2 \mathbf{S}(1) \rangle  \right] = \mathbb{E} \left[\left[ \frac{d}{ds} \log \tr \exp( \langle  \mathbf{h} + \mathbf{b}(1) + s \mathbf{w}, 2 \mathbf{S}(1) \rangle )  \right]_{s = 0}\right] \\ &= \langle \nabla \Lambda_\mathfrak{b}(\mathbf{h}), \mathbf{w} \rangle,
 \end{align*}
 where the second identity is an elementary fact and the last equality is due to dominated convergence and the uniform boundedness of the derivative. 
 For higher powers $d > 1$, the main observation is that for distinct integers $n_1, \ldots, n_d$ we have the identity
 $$   \tr \left[ \varrho_{\mathfrak{b}}(\mathbf{h}) \prod_{k = 1}^{d} \langle \mathbf{w}, \mathbf{S}(n_k) \rangle  \right]  =   \prod_{k = 1}^{d} \tr \left[ \varrho_{\mathfrak{b}}(\mathbf{h})  \langle \mathbf{w}, \mathbf{S}(n_k) \rangle  \right]         $$
 due to the tensor product structure of $\varrho_{\mathfrak{b}}(\mathbf{h})$. The second ingredient is that pairwise different indices $n_1, \ldots, n_d$ form the leading combinatorial contribution for large $N$ since $d! \binom{N}{d}/N^{d} \to 1$ as $N \to \infty$. Hence,
 \begin{align*}
 	 \tr \left[ \varrho_{\mathfrak{b}}(\mathbf{h}) (\langle \mathbf{w}, \mathbf{M} \rangle )^d \right] &= \frac{2^d d!}{N^d} \sum_{n_1 < \cdots < n_d} \tr \left[ \varrho_{\mathfrak{b}}(\mathbf{h}) \prod_{k = 1}^{d} \langle \mathbf{w}, \mathbf{S}(n_k) \rangle  \right] + o(1) \\
 	&= \frac{2^d d!}{N^d} \sum_{n_1 < \cdots < n_d} \prod_{k = 1}^{d} \tr \left[ \varrho_{\mathfrak{b}}(\mathbf{h})  \langle \mathbf{w}, \mathbf{S}(n_k) \rangle  \right]    +o(1) \\
 	&= \left(\tr \left[ \varrho_{\mathfrak{b}}(\mathbf{h}) \langle \mathbf{w}, \mathbf{M} \rangle  \right]\right)^d +o(1),
 \end{align*}
 where all $o(1)$ are uniform in the disorder. By the $d =1$ convergence result, we obtain the desired almost sure convergence $\lim_{N \to \infty} \tr \left[ \varrho_{\mathfrak{b}}(\mathbf{h}) (\langle \mathbf{w}, \mathbf{M} \rangle )^d \right] =  (\langle \mathbf{w}, \mathbf{m} \rangle )^d$.
\end{proof}

With these preparations, we are ready to spell the proof of Proposition~\ref{prop:lower}.
\begin{proof}[Proof of Proposition~\ref{prop:lower}]
	Our starting point is the following algebraic fact: for any polynomial $\textrm{Q}$ there exist some unit vectors $\mathbf{w}_1, \ldots, \mathbf{w}_K$, real numbers $\alpha_1, \ldots, \alpha_K$ and natural numbers $d_1, \ldots, d_K$ such that 
	\begin{equation}
		\textrm{Q}(\mathbf{M}) := \sum_{k = 1}^{K} \alpha_k (\langle \mathbf{w}_k, \mathbf{M} \rangle )^{d_k}
	\end{equation}
	holds as operator identity and also as identity for functions on $\rr^3$ \cite[Lemma A.6]{MW23}. Hence, Lemma~\ref{lem:lower} implies the almost sure convergence
	$$ 	\lim_{N \to \infty} \tr \left[ \varrho_{\mathfrak{b}}(\mathbf{h}) \textrm{Q}( \mathbf{M})  \rangle )^d \right] = \textrm{Q}( \mathbf{m}) $$ 
	with $\mathbf{m} = \nabla \Lambda_\mathfrak{b}(\mathbf{h})$. We now employ Gibbs' variational principle \eqref{eq:Gibbs} and obtain with $\mathbf{m} = \nabla \Lambda_\mathfrak{b}(\mathbf{h})$
	\begin{align*}
		p_N(\textrm{P},\mathfrak{b}) &\geq \frac1N \left(\tr[\varrho_{\mathfrak{b}}(\mathbf{h}) H_{\textrm{P},\mathfrak{b}}] - \tr[ \varrho_{\mathfrak{b}}(\mathbf{h})\log \varrho_{\mathfrak{b}}(\mathbf{h})] \right)\\
		&= \tr[ \varrho_{\mathfrak{b}}(\mathbf{h}) (\textrm{P}(\mathbf{M}) - \langle \mathbf{h}, \mathbf{M} \rangle ) ] + p_N(\langle \mathbf{h}, \mathbf{x} \rangle ,\mathfrak{b}) \\
		&\overset{N \to \infty}{\to} \textrm{P}(\mathbf{m}) - \langle \mathbf{h}, \mathbf{m} \rangle + \Lambda_{\mathfrak{b}}(\mathbf{h}) = \textrm{P}(\mathbf{m}) - \Lambda_{\mathfrak{b}}^{*}(\mathbf{m}) \quad \text{a.s.}
	\end{align*}
	For the second line, we recall that $\varrho_{\mathfrak{b}}(\mathbf{h})$ is the Gibbs state of the Hamiltonian $H_{\mathfrak{b}} + N \langle \mathbf{h}, \mathbf{M}\rangle$ and the last line is due to our previous result on tracial expectations for polynomials, the almost sure convergence of the pressure for linear Hamiltonian and Proposition~\ref{lem:lambda}. Since this holds true for any $\mathbf{h} \in \rr^3$, we obtain in view of Lemma~\ref{lem:lambda}
	$$ \sup_{ \mathbf{m} \in \mathring{B}_{\rr^3}} \left(\textrm{P}(\mathbf{m}) - \Lambda^{*}_\mathfrak{b}(\mathbf{m}) \right) \leq 	\liminf_{N \to \infty} p_N(\textrm{P},\mathfrak{b}). $$
	By continuity, the lower bound extends to the whole unit ball $B_{\rr^3}$.
\end{proof}

\subsection{Control of Thermal Fluctuations by Random Perturbations}\label{sec:therm}

We start by introducing some further standard notions. Given an Hamiltonian $H$ on $\mathcal{H}_N$, we denote by $p_N(H) = \frac1N \log \tr e^{H}$ the corresponding specific pressure and by $Z_N(H) := \tr e^{H}$ the respective partition function.  Moreover for an observable $A$, i.e. a self-adjoint operator on $\mathcal{H}_N$ we introduce the  thermal average 
\begin{equation}\label{eq:Gibbsav}
	\langle A \rangle_{H} := \frac{\tr e^{H} A}{\tr e^{H}} = \tr \varrho_H A
\end{equation}
and we denote by $\mathring{A} := A - \langle A \rangle_{H}$ the 'centered' observable. The variance $\langle \mathring{A}^2 \rangle_H$ measures the thermal fluctuations of $A$. In view of the proof of Proposition~\ref{prop:lower}, it is tempting to show a matching upper bound by showing that the magnetization operator $\mathbf{M}$ concentrates around its thermal average. However, this is not even true for the classical Curie-Weiss model in the low temperature phase. The problem is that the limiting Gibbs state might consists of multiple pure states and one would like to select on of the pure states to carry on the analysis. An elegant approach originated in spin glass theory is to let a further source of randomness pick (randomly) such a pure state. Indeed, we consider the perturbed Hamiltonian
\begin{equation}\label{eq:perturb}
	H_{\textrm{P}, \mathfrak{b}}(\pmb{\gamma}) = 	H_{\textrm{P}, \mathfrak{b}} + \sqrt{N} \langle \pmb{\gamma}, \mathbf{M} \rangle,
\end{equation}
where $\gamma_x, \gamma_y, \gamma_z$ are i.i.d. standard Gaussians mutually independent of the random fields $\{\mathbf{b}(n)\}_{n = 1}^{\infty}$.
The scaling with $\sqrt{N}$ chosen to be big enough to enforce thermal concentration, yet small enough to effect the pressure only on a sub-leading order. The latter claim is simple and recorded for future purposes in the following 
\begin{lemma}\label{lem:gauss}
	For any realization of the random field $\{\mathbf{b}(n)\}_{n = 1}^{\infty}$,
	\begin{equation}
		p_N(H_{\textrm{P}, \mathfrak{b}}) \leq \mathbb{E}_{\pmb{\gamma}}[p_N(H_{\textrm{P}, \mathfrak{b}}(\pmb{\gamma}))] \leq p_N(H_{\textrm{P}, \mathfrak{b}}) + \frac{C}{\sqrt{N}},
	\end{equation}
	where $\mathbb{E}_{\pmb{\gamma}}$ denotes the disorder average with respect to $\pmb{\gamma}$ and $C > 0$ is some universal constant.
\end{lemma}
\begin{proof}
	The lower bound follows from Jensen's inequality and for the upper bound we have the pointwise operator norm bound $p_N(H_{\textrm{P}, \mathfrak{b}}(\pmb{\gamma})) \leq p_N(H_{\textrm{P}, \mathfrak{b}}) + \frac{|\pmb{\gamma}|}{\sqrt{N}}$ from which the claim follows after taking expectations.
\end{proof}

The following proposition establish the concentration of the magnetization vector and is one main ingredient for our proof.
\begin{proposition}\label{prop:fluc}
	For any realization of the random field $\{\mathbf{b}(n)\}_{n = 1}^{\infty}$, the (averaged) thermal fluctuations of $\mathbf{M}$ is bounded by
	\begin{equation}\label{eq:fluc}
  \sum_{\mu = x,y,z}	\mathbb{E}_{\pmb{\gamma}}[\langle \mathring{\mathbf{M}}^2_\mu \rangle_{H_{\textrm{P}, \mathfrak{b}}(\pmb{\gamma})}] \leq  ( C_P + 6 \, \overline{\mathbf{b}}^{2/3}) N^{-1/3},
	\end{equation}
	with a constant $C_P$ which only depends on the polynomial $\textrm{P}$  and the abbreviation $\overline{\mathbf{b}}^2 = \left( \sum_{n = 1}^{N} \frac1N |\mathbf{b}(n)|\right)^2$.
\end{proposition} 

The bound \eqref{eq:fluc} is most likely not sharp and we expect a central limit theorem (CLT) on the scale of $N^{-1/2}$. Such a result would be very interesting, but this has not even been established for the deterministic model $H_{\textrm{P}}$. Note that the $\mathbf{b}$-expectation of the term $\overline{\mathbf{b}}_{\mu}^{2/3}$ is uniformly bounded by $\mathbb{E}[|\mathfrak{b}|]$.

Our proof follows the strategy in \cite{STT20}, where an analogous result has (implicitly) been derived for the Quantum Hopfield Model. Our main observation is that while the main line of the proof in \cite{STT20} is tailored to quadratic Hamiltonians, the concentration argument can be lifted to more general models with only slight complications due to a higher degree of non-commutativity and the random external field. 

We start from the exact diagonalization of the Hamiltonian $H_{\textrm{P}, \mathfrak{b}}(\pmb{\gamma})$. There exist an orthonormal basis of eigenvectors $\{\psi_k\}_{k = 1}^{2^N} \subset \mathcal{H}_N$ such that 
	$$ H_{\textrm{P}, \mathfrak{b}}(\pmb{\gamma}) \psi_k = E_k \psi_k $$
	with the corresponding real eigenvalues $\{E_k\}_{k = 1}^{2^N}$. Of course, $\psi_k, E_k$ depend on the further parameters $\textrm{P}, \mathfrak{b}, \pmb{\gamma}$, but we drop this dependencies from our notation. Moreover, we set $M_{k,l}^{\mu} := \langle \psi_k, \mathring{\mathbf{M}}_\mu \psi_l \rangle$. 
	We need a few spectral identities.
	\begin{lemma}\label{lem:iden}
		\begin{enumerate}
			\item  \begin{equation}\label{eq:specav} 
				\langle \mathring{\mathbf{M}}^2_\mu \rangle_{H_{\textrm{P}, \mathfrak{b}}(\pmb{\gamma})} = \frac{1}{2 Z_N(H_{\textrm{P}, \mathfrak{b}}(\pmb{\gamma}))} \sum_{k,l = 1}^{2^N} |M_{k,l}^{\mu}|^2 (e^{E_k} + e^{E_l})
			\end{equation}
			\item  \vspace{-0.5cm} \begin{equation}\label{eq:speccomm}
				\langle [\mathbf{M}_\mu, H_{\textrm{P}, \mathfrak{b}}(\pmb{\gamma})]^2 \rangle_{H_{\textrm{P}, \mathfrak{b}}(\pmb{\gamma})} = -\frac{1}{2 Z_N(H_{\textrm{P}, \mathfrak{b}}(\pmb{\gamma}))} \sum_{k,l = 1}^{2^N} |M_{k,l}^{\mu}|^2 (e^{E_k} + e^{E_l})|E_k - E_l|^2,
			\end{equation}
			where $[A,B] := AB - BA$ denotes the commutator.
			\item  \vspace{-0.3cm} \begin{equation}\label{eq:specderiv}
				\frac{\partial^2}{\partial \gamma_u^2} p_N(H_{\textrm{P}, \mathfrak{b}}(\pmb{\gamma})) = \frac{1}{ Z_N(H_{\textrm{P}, \mathfrak{b}}(\pmb{\gamma}))} \sum_{k,l = 1}^{2^N} |M_{k,l}^{\mu}|^2 \frac{e^{E_k} - e^{E_l}}{E_k - E_l},
			\end{equation}
			where $\frac{e^{E_k} - e^{E_l}}{E_k - E_l}$ is interpreted as $e^{E_k}$ if $E_k = E_l$.
		\end{enumerate}
	\end{lemma}
	
	These identities are well known and the proofs are based on direct computations. For the reader's convenience, we present a proof in the appendix. We need a further elementary bound.
	\begin{lemma}\label{lem:cosh}
		For any $x > 0$,
		\begin{equation}
			\cosh x \leq \sinh x + \frac{\sinh x}{x}
		\end{equation}
	\end{lemma}
	The proof of Lemma~\ref{lem:cosh} can also be found in the appendix and we dive now into the proof of Proposition~\ref{prop:fluc}. The main idea is to relate the thermal fluctuations $\langle \mathring{\mathbf{M}}^2_\mu \rangle_{H_{\textrm{P}, \mathfrak{b}}(\pmb{\gamma})}$ to the Gaussian derivative $\frac{\partial^2}{\partial \gamma_u^2} p_N(H_{\textrm{P}, \mathfrak{b}}(\pmb{\gamma}))$. While in the classical case one obtains an exact identity, it was observed in \cite{STT20} that one can still control the thermal fluctuations by $\frac{\partial^2}{\partial \gamma_u^2} p_N(H_{\textrm{P}, \mathfrak{b}}(\pmb{\gamma}))$ and some commutator bounds, which are more involved in our case.
	
	\begin{proof}[Proof of Proposition~\ref{prop:fluc}]
		We start from the identity \eqref{eq:specav} and estimate
		\begin{align*}
			\langle \mathring{\mathbf{M}}^2_\mu \rangle_{H_{\textrm{P}, \mathfrak{b}}(\pmb{\gamma})} &= \frac{1}{2 Z_N(H_{\textrm{P}, \mathfrak{b}}(\pmb{\gamma}))} \sum_{k,l = 1}^{2^N} |M_{k,l}^{\mu}|^2 (e^{E_k} + e^{E_l}) \\
			&\leq  \frac{1}{Z_N(H_{\textrm{P}, \mathfrak{b}}(\pmb{\gamma}))} \sum_{k,l = 1}^{2^N} |M_{k,l}^{\mu}|^2 \left(\frac{e^{E_k} - e^{E_l}}{E_k - E_l}) + \frac12 |e^{E_k} - e^{E_l}| \right) \\
			&= \frac{\partial^2}{\partial \gamma_u^2} p_N(H_{\textrm{P}, \mathfrak{b}}(\pmb{\gamma})) + \frac{1}{2 Z_N(H_{\textrm{P}, \mathfrak{b}}(\pmb{\gamma}))} \sum_{k,l = 1}^{2^N} |M_{k,l}^{\mu}|^2 |e^{E_k} - e^{E_l}| =: \frac{\partial^2}{\partial \gamma_u^2} p_N(H_{\textrm{P}, \mathfrak{b}}(\pmb{\gamma})) + R,
		\end{align*}
		where the fraction $\frac{e^{E_k} - e^{E_l}}{E_k - E_l}$ is extended continuously for the case $E_k = E_l$. For the second line, we write $e^{E_k} + e^{E_l} = e^{(E_k+E_l)/2}(e^{(E_k-E_l)/2} + e^{(E_l-E_k)/2}$ and apply Lemma~\ref{lem:cosh}. The last line follows from Lemma~\ref{lem:iden}. The remainder term $R$ is estimated via Hölder's inequality 
		\begin{align*}
			R &= \frac{1}{2 Z_N(H_{\textrm{P}, \mathfrak{b}}(\pmb{\gamma}))} \sum_{k,l = 1}^{2^N} |M_{k,l}^{\mu}|^2 \frac{|e^{E_k} - e^{E_l}|}{|E_k - E_l|} |E_k - E_l| \\
			& \leq \left(\frac{1}{2Z_N(H_{\textrm{P}, \mathfrak{b}}(\pmb{\gamma}))} \sum_{k,l = 1}^{2^N} |M_{k,l}^{\mu}|^2 \frac{e^{E_k} - e^{E_l}}{E_k - E_l} \right)^{2/3} \left(\frac{1}{2 Z_N(H_{\textrm{P}, \mathfrak{b}}(\pmb{\gamma}))} \sum_{k,l = 1}^{2^N} |M_{k,l}^{\mu}|^2 (e^{E_k} + e^{E_l})|E_k - E_l|^2 \right)^{1/3} \\
			&= \left( \frac12 	\frac{\partial^2}{\partial \gamma_u^2} p_N(H_{\textrm{P}, \mathfrak{b}}(\pmb{\gamma}))  \right)^{2/3} \left( -	\langle [\mathbf{M}_\mu, H_{\textrm{P}, \mathfrak{b}}(\pmb{\gamma})]^2 \rangle_{H_{\textrm{P}, \mathfrak{b}}(\pmb{\gamma})}\right)^{1/3},
		\end{align*}
		where we used the identities \eqref{eq:speccomm} and \eqref{eq:specderiv} for the last line. As next step, we estimate the contribution of the commutator $[\mathbf{M}_\mu, H_{\textrm{P}, \mathfrak{b}}(\pmb{\gamma})]$. We employ a crude operator norm estimate
		\begin{align*}
			|\langle [\mathbf{M}_\mu, H_{\textrm{P}, \mathfrak{b}}(\pmb{\gamma})]^2 \rangle_{H_{\textrm{P}, \mathfrak{b}}(\pmb{\gamma})}| \leq \|[\mathbf{M}_\mu, H_{\textrm{P}, \mathfrak{b}}(\pmb{\gamma})] \|^2 \leq 3( \|[\mathbf{M}_\mu, H_{\textrm{P}}] \|^2 + \|[\mathbf{M}_\mu, H_{ \mathfrak{b}}] \|^2 + \|[\mathbf{M}_\mu, \sqrt{N} \langle \pmb{\gamma}, \mathbf{M} \rangle] \|^2),
		\end{align*}
		and estimate each term separately. Recall the canonical commutation relations of spin operators, i.e.,
		$[s_x,s_y] = s_z$ etc., which  in particular  imply the estimate $\|[\mathbf{M}_\mu, \mathbf{M}_\nu] |\ \leq \frac{2}{N}$. Hence, the following bounds are readily derived
		\begin{align*}
			\|[\mathbf{M}_\mu, H_{ \mathfrak{b}}] \|^2 \leq 2 \left( \frac1N \sum_{n = 1}^{N} |\mathbf{b}(n)| \right)^2, \qquad \|[\mathbf{M}_\mu, \sqrt{N} \langle \pmb{\gamma}, \mathbf{M} \rangle] \|^2) \leq \frac{2 |\pmb{\gamma}|^2}{N}.
		\end{align*}
		Bounding the commutator of the polynomial part is slightly more complicated. It is convenient to use again the decomposition $\textrm{P}(\mathbf{M}) = \sum_{k = 1}^{K} \alpha_k (\langle \mathbf{w}_k, \mathbf{M} \rangle )^{d_k}$ such that in view of the AM-GM inequality it is enough to consider a single monomial $\langle \mathbf{w}, \mathbf{M} \rangle )^{d}$. Using the relation $[\langle \mathbf{w}, \mathbf{M} \rangle )^{d}, \mathbf{M}_\mu] = \langle \mathbf{w}, \mathbf{M} \rangle [\langle \mathbf{w}, \mathbf{M} \rangle )^{d-1}, \mathbf{M}_\mu] + [\langle \mathbf{w}, \mathbf{M} \rangle )^{d-1}, \mathbf{M}_\mu] \langle \mathbf{w}, \mathbf{M} \rangle$, we deduce 
		$$  \|[\mathbf{M}_\mu, \langle \mathbf{w}, \mathbf{M} \rangle )^{d}] \| \leq \frac{C(d)}{N} $$
		with some constant $C(d)$ which only depends on the degree $d$- Combining the contribution of the monomials, we obtain 
		$$ \|[\mathbf{M}_\mu, H_{\textrm{P}}] \|^2  \leq C_P$$
		with some constant $C_P$. The last ingredient is the bound
		\begin{align*}
			\mathbb{E}_{\pmb{\gamma}}\left[\frac{\partial^2}{\partial \gamma_u^2} p_N(H_{\textrm{P}, \mathfrak{b}}(\pmb{\gamma})) \right] = \mathbb{E}_{\pmb{\gamma}}\left[\gamma_u \frac{\partial}{\partial \gamma_u} p_N(H_{\textrm{P}, \mathfrak{b}}(\pmb{\gamma})) \right] =   \mathbb{E}_{\pmb{\gamma}}\left[\sqrt{N} \gamma_u \langle \mathbf{M}_\mu \rangle_{H_{\textrm{P}, \mathfrak{b}}(\pmb{\gamma})} \right] \leq C \sqrt{N},
		\end{align*}
		where the first equality is due to Gaussian partial integration. We now combine all bounds, 
		\begin{align*}
			\mathbb{E}_{\pmb{\gamma}}\left[\langle \mathring{\mathbf{M}}^2_\mu \rangle_{H_{\textrm{P}, \mathfrak{b}}(\pmb{\gamma})} \right] &= \mathbb{E}_{\pmb{\gamma}}\left[\frac{\partial^2}{\partial \gamma_u^2} p_N(H_{\textrm{P}, \mathfrak{b}}(\pmb{\gamma})) \right] + \mathbb{E}_{\pmb{\gamma}}[R] \\
			&\leq \frac{C}{\sqrt{N}} + \mathbb{E}_{\pmb{\gamma}}\left[ \left( \frac12 	\frac{\partial^2}{\partial \gamma_u^2} p_N(H_{\textrm{P}, \mathfrak{b}}(\pmb{\gamma}))  \right)^{2/3} \left( -	\langle [\mathbf{M}_\mu, H_{\textrm{P}, \mathfrak{b}}(\pmb{\gamma})]^2 \rangle_{H_{\textrm{P}, \mathfrak{b}}(\pmb{\gamma})}\right)^{1/3} \right] \\
			& \leq \frac{C}{\sqrt{N}} + \left( \frac12 \mathbb{E}_{\pmb{\gamma}}\left[	\frac{\partial^2}{\partial \gamma_u^2} p_N(H_{\textrm{P}, \mathfrak{b}}(\pmb{\gamma})) \right] \right)^{2/3} \left( - \mathbb{E}_{\pmb{\gamma}}\left[	\langle [\mathbf{M}_\mu, H_{\textrm{P}, \mathfrak{b}}(\pmb{\gamma})]^2 \rangle_{H_{\textrm{P}, \mathfrak{b}}(\pmb{\gamma})}\right] \right)^{1/3}		\\
			& \leq \frac{C}{\sqrt{N}} + CN^{-1/3}\left(C_P + 6/N + 2 \left( \frac1N \sum_{n = 1}^{N} |\mathbf{b}(n)| \right)^2\right)^{1/3},
		\end{align*}
		where we used Hölder's inequality in the third line. After adding up all components, using the concavity of the function $x \mapsto x^{1/3}$ and redefining the constant $C_P$, the claimed bound follows.
	\end{proof}
	
	\subsection{Negative Definite Quadratic Models via Linearization}
	
	In this section, we restrict ourselves to quadratic Hamiltonians with symbol $\textrm{P}_{\pmb{\alpha}}(\mathbf{m}) = -( \alpha_x m_x^2 + \alpha_y m_y^2 + \alpha_z m_z^2)$ with nonnegative $\alpha_x, \alpha_y, \alpha_z \geq 0$. The total Hamiltonian has the form
	\begin{equation}\label{eq:quadr}
		H_{\textrm{P}_{\pmb{\alpha}}, \mathfrak{b}} = - \frac{4}{N} \left(\alpha_x \mathbf{S}_x^2 + \alpha_y \mathbf{S}_y^2 + \alpha_z \mathbf{S}_z^2 \right) + 2\sum_{n = 1}^{N} \langle \mathbf{b}(n), \mathbf{S}(n) \rangle.
	\end{equation}
	acting on the $N$-particle Hilbert space $\mathcal{H}_N$. Here, $\mathbf{S}$ denotes as before the total spin operator. The class of Hamiltonians in \eqref{eq:quadr} is a very special case of our general setting, but it turns out to be key to understand these special Hamiltonians first before turning to more general polynomials. Besides the quadratic nature, also the negative definiteness (i.e. $\pmb{\alpha} \geq 0$ componentwise) is crucial for the discussion in this section.
	
	The main idea is to approximate the thermodynamics of quadratic models by an "optimal" choice of a linear Hamiltonian $	L(\pmb{\gamma}, \pmb{m}, \{\mathbf{b}(n)\}_{n = 1}^{N})$,
	\begin{align}\label{eq:lin}
			L(\pmb{\alpha}, \pmb{m},& \{\mathbf{b}(n)\}_{n = 1}^{N}) := H_{\textrm{P}_{\pmb{\alpha}}, \mathfrak{b}} + N \big(\alpha_x (\mathbf{M}_x -m_x)^2 + \alpha_y (\mathbf{M}_y -m_y)^2 + \alpha_z (\mathbf{M}_z -m_z)^2 \big) \nonumber \\
			&= N(\alpha_x m_x^2 + \alpha_y m_y^2 + \alpha_z m_z^2) - 4 \left(\alpha_x m_x \mathbf{S}_x + \alpha_y m_y \mathbf{S}_y + \alpha_z m_z \mathbf{S}_z \right) + 2 \sum_{n = 1}^{N} \langle \mathbf{b}(n), \mathbf{S}(n) \rangle,
	\end{align}
	where $\mathbf{m} \in \rr^3$ is at first a free parameter, but as the notation suggest will eventually be chosen to essentially agree with thermal average $\langle \mathbf{M} \rangle_{H_{\textrm{P}_{\pmb{\alpha}}, \mathfrak{b}}}$. We record two important observations. 
	\begin{enumerate}
		\item Since we obtain $L(\pmb{\alpha}, \pmb{m}, \{\mathbf{b}(n)\}_{n = 1}^{N})$ by adding a positive definite operator to the original Hamiltonian $H_{\textrm{P}_{\pmb{\alpha}}, \mathfrak{b}}$, we have the operator inequality $ H_{\textrm{P}_{\pmb{\alpha}}, \mathfrak{b}}  \leq L(\pmb{\alpha}, \pmb{m},\{\mathbf{b}(n)\}_{n = 1}^{N}). $ This immediately implies the corresponding inequality for the pressure 
		\begin{equation}\label{eq:inpress}
			p_N(H_{\textrm{P}_{\pmb{\alpha}}, \mathfrak{b}}) \leq \inf_{\mathbf{m} \in B_{\rr^3}} p_N(L(\pmb{\alpha}, \pmb{m}, \{\mathbf{b}(n)\}_{n = 1}^{N})).
		\end{equation}
		Note that \eqref{eq:inpress} crucially depends on the fact that $\pmb{\alpha} \geq 0$ and the following discussion fails for more general quadratic Hamiltonians.
		\item On the other hand, $L(\pmb{\alpha}, \pmb{m}, \{\mathbf{b}(n)\}_{n = 1}^{N})$ is a linear Hamiltonian in the sense that it depends only linearly on the spin operators $\mathbf{S}(n)$. Hence, $L(\pmb{\alpha}, \pmb{m}, \{\mathbf{b}(n)\}_{n = 1}^{N})$  is exactly diagonalizable and its pressure is given by
		\begin{align}\label{eq:freeL}
			p_N(L(\pmb{\alpha}, \pmb{m}, &\{\mathbf{b}(n)\}_{n = 1}^{N})) := \alpha_x m_x^2 + \alpha_y m_y^2 + \alpha_z m_z^2) \nonumber \\ &+ \frac1N \sum_{n = 1}^{N} \log 2 \cosh \sqrt{(-2\alpha_x m_x + b_x(n))^2 + (-2\alpha_y m_y + b_y(n))^2 + (-2\alpha_z m_z + b_z(n))^2}
		\end{align}
		Employing the strong law of large numbers, we obtain the almost sure convergence 
		\begin{equation}\label{eq:pconv}
			\lim_{N \to \infty} p_N(L(\pmb{\alpha}, \pmb{m}, \{\mathbf{b}(n)\}_{n = 1}^{N})) = -\textrm{P}_{\pmb{\alpha}}(\mathbf{m}) + \Lambda_{\mathfrak{b}}(-2 \pmb{\alpha} \odot \pmb{m} )
		\end{equation}
		with the componentwise or Hadamard product $\odot$.
	\end{enumerate}
	As next step, we want to extend the almost sure convergence in \eqref{eq:pconv} to  $\inf_{\mathbf{m} \in B_{\rr^3}} p_N(L(\pmb{\alpha}, \pmb{m}, \{\mathbf{b}(n)\}_{n = 1}^{N}))$. This is content of the following
	\begin{lemma}\label{lem:conv}
		We have the almost sure convergence
		\begin{equation}
			\lim_{N \to \infty} \inf_{\mathbf{m} \in B_{\rr^3}} p_N(L(\pmb{\alpha}, \pmb{m}, \{\mathbf{b}(n)\}_{n = 1}^{N})) = \inf_{\mathbf{m} \in B_{\rr^3}} \bigg( -\textrm{P}_{\pmb{\alpha}}(\mathbf{m}) + \Lambda_{\mathfrak{b}}(-2 \pmb{\alpha} \odot \pmb{m} ) \bigg)
		\end{equation}
	\end{lemma}
	
	The proof of Lemma~\ref{lem:conv} exploits the uniform Lipschitz continuity in $\textbf{m}$. The straightforward argument is deferred to the appendix. We continue with our analysis of $	p_N(H_{\textrm{P}_{\pmb{\alpha}}, \mathfrak{b}})$. We want to show that we the bound from \eqref{eq:inpress} becomes an equality in the thermodynamical limit. To this end we employ a version of the Boguliobov inequality, namely for any self-adjoint matrices $H$ and $A$ on $\mathcal{H}_N$
	\begin{equation}\label{eq:bogu}
		p_N(H) \leq p_N(H+A) - \frac1N \langle A \rangle_{H}.
	\end{equation}
	In order to employ the Boguliobov inequality \eqref{eq:bogu} and Proposition~\ref{prop:fluc}, we perturb by a random Gaussian field Hamiltonian and obtain
	\begin{align*}
		 \inf_{\mathbf{m} \in B_{\rr^3}} p_N(L(\pmb{\alpha}, \pmb{m}, &\{\mathbf{b}(n)\}_{n = 1}^{N})) - p_N(H_{\textrm{P}_{\pmb{\alpha}}, \mathfrak{b}}) \\ &\leq \frac{C}{\sqrt{N}} + \mathbb{E}{\pmb{\gamma}} \left[\inf_{\mathbf{m} \in B_{\rr^3}} p_N(L(\pmb{\alpha}, \pmb{m}, \{\mathbf{b}(n) + N^{-1/2} \pmb{\gamma} \}_{n = 1}^{N})) - p_N(H_{\textrm{P}_{\pmb{\alpha}}, \mathfrak{b}}(\pmb{\gamma})) \right] \\
		 & \leq \frac{C}{\sqrt{N}} + \mathbb{E}{\pmb{\gamma}} \left[ p_N(L(\pmb{\alpha}, \langle \pmb{M} \rangle_{H_{\textrm{P}, \mathfrak{b}}(\pmb{\gamma})}, \{\mathbf{b}(n) + N^{-1/2} \pmb{\gamma} \}_{n = 1}^{N})) - p_N(H_{\textrm{P}_{\pmb{\alpha}}, \mathfrak{b}}(\pmb{\gamma})) \right] \\
		 & \leq \frac{C}{\sqrt{N}} + \sum_{\mu = x,y,z} \alpha_{\mu} \mathbb{E}_{\pmb{\gamma}}[\langle \mathring{\mathbf{M}}^2_\mu \rangle_{H_{\textrm{P}, \mathfrak{b}}(\pmb{\gamma})}]
		  \leq \frac{C}{\sqrt{N}} + \max_{\mu = x,y,z} \alpha_\mu \left( \leq  ( C_P + 6 \, \overline{\mathbf{b}}^{2/3}) N^{-1/3} \right).
	\end{align*}
	The first bound is due to Lemma~\ref{lem:gauss}, the second inequality is immediate and the third line follows from \eqref{eq:bogu} and Proposition~\ref{prop:fluc}. The last expression converges almost surely to zero and together with Lemma~\ref{lem:conv} we arrive at the following result.
	
	\begin{proposition}\label{prop:quadr}
		Let $\textrm{P}_{\pmb{\alpha}}(\mathbf{m}) = -( \alpha_x m_x^2 + \alpha_y m_y^2 + \alpha_z m_z^2)$ be a quadratic polynomial with $\pmb{\alpha} \geq 0$.
		Then,
		\begin{equation}\label{eq:quad}
				\lim_{N \to \infty} p_N(H_{\textrm{P}_{\pmb{\alpha}}, \mathfrak{b}}) = \inf_{\mathbf{m} \in B_{\rr^3}} \bigg( -\textrm{P}_{\pmb{\alpha}}(\mathbf{m}) + \Lambda_{\mathfrak{b}}(-2 \pmb{\alpha} \odot \pmb{m} ) \bigg) \quad \text{a.s.}
			\end{equation}
	\end{proposition}
	Note that the variational expression \eqref{eq:quad} is not identical with the formula \eqref{eq:main} of the main theorem. This apparent conflict can be resolved by invoking the Toland-Singer duality \cite{BC17}. This is only applicable since $\textrm{P}_{\pmb{\alpha}}$ is a concave function and hence one cannot expect an analogue of \eqref{eq:quad} for more general polynomials $\textrm{P}$. For future reference, we record the dual identity (i.e. \eqref{eq:maindetermin} in this special case)
	\begin{equation}\label{eq:quadr1}
		\lim_{N \to \infty} p_N(H_{\textrm{P}_{\pmb{\alpha}}, \mathfrak{b}}) = \sup_{\mathbf{m} \in B_{\rr^3}} \bigg( \textrm{P}_{\pmb{\alpha}}(\mathbf{m}) - \Lambda_{\mathfrak{b}}^{*}( \pmb{m} ) \bigg) \quad \text{a.s.}
	\end{equation}
	
	\subsection{Upper Bound via a Microcanonical Analysis}\label{sec:upper}
	
	The main goal of this Section to complement Proposition~\ref{prop:lower} by the corresponding sharp upper bound.
	\begin{proposition}\label{prop:upper}
		Let $\textrm{P} \colon \rr^3 \to \rr $ be a polynomial, $\mathfrak{b}$ an integrable vector valued random variable and $H_{\textrm{P},\mathfrak{b}}$ the associated random-field Hamiltonian(s) with specific pressures $p_N(\textrm{P},\mathfrak{b})$. Then, almost surely
		\begin{equation}\label{eq:upper}
			\sup_{\mathbf{m} \in B_{\rr^3}} \left(\textrm{P}(\mathbf{m}) - \Lambda^{*}_\mathfrak{b}(\mathbf{m}) \right) \geq 	\limsup_{N \to \infty} p_N(\textrm{P},\mathfrak{b})
		\end{equation}
	\end{proposition}
	
	As announced, we proceed by a variant of a microcanonical analysis. For classical systems, the strategy is very simple: one considers configurations with magnetization $\pmb{M} \simeq \pmb{m}$ and via a large-deviation analysis one is typically able to compute the restricted free energy. In the non-commutative setting there are different non-equivalent methods to restrict a Hamiltonian to a certain part of the phase space. The most intuitive approach via orthogonal projections is problematic as it requires good control of the Hamiltonian. It turns out that the most elegant strategy is to perturb the Hamiltonian by a potential which is very negative if $\pmb{M} \not\simeq \pmb{m}$. Our starting point is the
	following simple observation.
	
	\begin{lemma}\label{lem:micro}
			Let $\textrm{P} \colon \rr^3 \to \rr $ be a polynomial. Then there exists a function $g_{\textrm{P} } : [0,\infty] \to [0, \infty]$ such that $\lim_{\alpha \to \infty} g(\alpha) = 0$
			and 
			\begin{equation}\label{eq:micro}
				\textrm{P}(\mathbf{M}) \leq (	\textrm{P}(\mathbf{m}) + g_P(\alpha)) \mathbbm{1} + \alpha \left((\mathbf{M}_x - \mathbf{m}_x)^2 + (\mathbf{M}_y - \mathbf{m}_y)^2 + (\mathbf{M}_z - \mathbf{m}_z)^2
				\right)
			\end{equation}
			holds as operator inequality for all $\mathbf{m} \in B_{\rr^3}$ and $N \in \nn$.
	\end{lemma}
	
	\begin{proof}
		We use yet again the decomposition  $\textrm{P}(\mathbf{M}) =  \sum_{k = 1}^{K} \beta_k (\langle \mathbf{w}_k, \mathbf{M} \rangle )^{d_k}$ from which it follows that is enough to prove \eqref{eq:micro} for an arbitrary monomial $ \beta \langle \mathbf{w}, \mathbf{M} \rangle^{d}$. We keep the factor of $\beta$ to cover both possible signs. We claim that it is enough to show 
		\begin{equation}\label{eq:lemmicro}
		\beta	\langle \mathbf{w}, \mathbf{M} \rangle^{d} \leq \ (\beta \langle \mathbf{w}, \mathbf{m} \rangle^{d} +g(\alpha)) \mathbbm{1} + \alpha (\langle \mathbf{w}, \mathbf{M} \rangle - \langle \mathbf{w}, \mathbf{m} \rangle )^2 
		\end{equation}
		uniformly in $\mathbf{m}, N$ for a function $g$ with the claimed properties. Indeed, we have the operator identity
		$$ (\mathbf{M}_x - \mathbf{m}_x)^2 + (\mathbf{M}_y - \mathbf{m}_y)^2 + (\mathbf{M}_z - \mathbf{m}_z)^2 = (\langle \mathbf{e}_1, \mathbf{M} \rangle  - \langle \mathbf{e}_1, \mathbf{m} \rangle)^2 + (\langle \mathbf{e}_2, \mathbf{M} \rangle  - \langle \mathbf{e}_2, \mathbf{m} \rangle)^2 + (\langle \mathbf{e}_3, \mathbf{M} \rangle  - \langle \mathbf{e}_3, \mathbf{m} \rangle)^2 $$
		for any orthonormal basis $\mathbf{e}_1, \mathbf{e}_2,\mathbf{e}_3$ in $\rr^3$, which follows by expanding the squares and exploiting the rotational invariance of the square of the total spin vector. Hence, $ (\langle \mathbf{w}, \mathbf{M} \rangle - \langle \mathbf{w}, \mathbf{m} \rangle )^2 \leq (\mathbf{M}_x - \mathbf{m}_x)^2 + (\mathbf{M}_y - \mathbf{m}_y)^2 + (\mathbf{M}_z - \mathbf{m}_z)^2$ and it suffices to show \eqref{eq:lemmicro}.  Note that both sides in \eqref{eq:lemmicro} are diagonal in the same basis, and hence it is enough to prove the bound
		$$ \beta x^d \geq \beta x_0^d - g(\alpha) - \alpha (x-x_0)^2 $$
		for all $x, x_0 \in [-1,1]$. To this end, we apply a standard Taylor expansion and obtain 
		\begin{align*} \beta x^d &= \beta x_0^d + \beta d x_0^{d-1}(x-x_0) + \beta \frac{d(d-1)}{2}  \xi^{d-2} (x-x_0)^2 \leq \beta x_0^d + \beta d x_0^{d-1}(x-x_0) + C (x-x_0)^2 \\
		& =  \beta x_0^d + (\beta d x_0^{d-1}(x-x_0) - (\alpha -C ) (x-x_0)^2) + \alpha (x-x_0)^2 
		\end{align*}
		with the abbreviation $C = |\beta| d(d-1)/2$. We can choose $g(\alpha)$ as the maximum of the function $(x,x_0) \mapsto \beta d x_0^{d-1}(x-x_0) - (\alpha -C ) (x-x_0)^2$ on $[-1,1]^2$ and the claimed property $\lim_{\alpha \to \infty} g(\alpha) = 0$ is easily checked.
		
	\end{proof}
	
	The main idea behind Lemma~\ref{lem:micro} is a comparison with quadratic models. The main trick is that we are able to counter balance the positive quadratic term in \eqref{eq:micro} and in parallel  derive a microcanonical expression.
	
	\begin{proof}[Proof of Proposition~\ref{prop:upper}]
		It is convenient to introduce the quadratic operator $Q_{\mathbf{m}} := N((\mathbf{M}_x - \mathbf{m}_x)^2 + (\mathbf{M}_y - \mathbf{m}_y)^2 + (\mathbf{M}_z - \mathbf{m}_z)^2)$ and we further recall the perturbed Hamiltonian $	H_{\textrm{P}, \mathfrak{b}}(\pmb{\gamma})$ from \eqref{eq:perturb}.
		For any fixed $\alpha > 0$,
		\begin{align}
			0 \leq \mathbb{E}_{\pmb{\gamma}}\left[ p_N(	H_{\textrm{P}, \mathfrak{b}}(\pmb{\gamma})) -p_N\left(	H_{\textrm{P}, \mathfrak{b}}(\pmb{\gamma}) - 2\alpha Q_{\langle \mathbf{M} \rangle_{H_{\textrm{P}, \mathfrak{b}}(\pmb{\gamma})}} \right) \right] \leq 2\alpha  \sum_{\mu = x,y,z}	\mathbb{E}_{\pmb{\gamma}}[\langle \mathring{\mathbf{M}}^2_\mu \rangle_{H_{\textrm{P}, \mathfrak{b}}(\pmb{\gamma})}],
		\end{align}
		where the last bound is due to the Boguliobov inequality \eqref{eq:bogu}. Note that Proposition~\ref{prop:fluc} implies that $\mathbb{E}_{\pmb{\gamma}}[\langle \mathring{\mathbf{M}}^2_\mu \rangle_{H_{\textrm{P}, \mathfrak{b}}(\pmb{\gamma})}]$ converges almost surely to zero. Since the Gaussian perturbation $\pmb{\gamma}$ does in expectation only shift the pressure by at most $\mathcal{O}(N^{-1/2})$ (see Lemma~\ref{lem:gauss}), we obtain the almost sure bound
		$$ \limsup_{N \to \infty} p_N(	H_{\textrm{P}, \mathfrak{b}}) \leq \limsup_{N \to \infty} \sup_{ \mathbf{m} \in B_{\rr^3}} p_N(	H_{\textrm{P}, \mathfrak{b}} - 2\alpha Q_{\mathbf{m}}) $$
		for every $\alpha > 0$. In a next step we employ the operator inequality from Lemma~\ref{lem:micro},
		$$  \limsup_{N \to \infty} p_N(	H_{\textrm{P}, \mathfrak{b}}) \leq  \limsup_{N \to \infty} \sup_{ \mathbf{m} \in B_{\rr^3}} \bigg(\textrm{P}(\mathbf{m}) + g_P(\alpha) + p_N( H_{\mathfrak{b}} - \alpha Q_{\mathbf{m}})  \bigg)    $$
		The Hamiltonian $- \alpha Q_{\mathbf{m}}$ is quadratic and negative definite. Hence, we have according to Proposition~\ref{prop:quadr} (or more precisely the dual identity \eqref{eq:quadr1}) the almost sure convergence
		\begin{align*}
			\lim_{N \to \infty} p_N( H_{\mathfrak{b}} - \alpha Q_{\mathbf{m}}) = \sup_{ \mathbf{m}' \in B_{\rr^3}} \bigg( -\alpha \left( (\mathbf{m}_x - \mathbf{m}'_x)^2 + (\mathbf{m}_y - \mathbf{m}'_y)^2  +  (\mathbf{m}_z - \mathbf{m}'_z)^2 \right) - \Lambda_{\mathfrak{b}}^{*}( \pmb{m}' ) \bigg).
		\end{align*}
	More precisely, we made also use of the fact that \eqref{eq:quadr1} implies \eqref{eq:main} for all negative definite quadratic polynomials plus a linear field, which can be proved by shifting the random fields $\mathbf{b}(n)$ by some constant field $\mathbf{h}$.
	We claim that the almost sure convergence extends to the whole variational expression in $\textbf{m}$, i.e.,
	\begin{align*}
		  \lim_{N \to \infty} &\sup_{ \mathbf{m} \in B_{\rr^3}} \bigg(\textrm{P}(\mathbf{m}) + g_P(\alpha) + p_N( H_{\mathfrak{b}} - \alpha Q_{\mathbf{m}})  \bigg)  \\
		= &\sup_{ \mathbf{m} \in B_{\rr^3}} \sup_{ \mathbf{m}' \in B_{\rr^3}} \bigg(\textrm{P}(\mathbf{m}) + g_P(\alpha) -\alpha \left( (\mathbf{m}_x - \mathbf{m}'_x)^2 + (\mathbf{m}_y - \mathbf{m}'_y)^2  +  (\mathbf{m}_z - \mathbf{m}'_z)^2 \right) - \Lambda_{\mathfrak{b}}^{*}( \pmb{m}') \bigg)  
	\end{align*} 
	This can be proved as Lemma~\ref{lem:conv}: due to Lemma~\ref{lem:convv} is to establish a uniform Lipschitz continuity of $\textrm{P}(\mathbf{m}) + g_P(\alpha) + p_N( H_{\mathfrak{b}} - \alpha Q_{\mathbf{m}})$ in $\mathbf{m}$. This is immediate for $\textrm{P}(\mathbf{m})$ and follows for the pressure $p_N( H_{\mathfrak{b}} - \alpha Q_{\mathbf{m}})$ from simple operator norm bound
	$$  \|Q_{\mathbf{m}_1} - Q_{\mathbf{m}_2} \| \leq 2N \|\mathbf{m}_1 - \mathbf{m}_2 \|. $$
	Since our bound holds true for every $\alpha > 0$ we may take the $\alpha \to \infty$ limit for which we recall that $g_P(\alpha) \to 0$ and therefore
	\begin{align*}
		\lim_{\alpha \to \infty} &\sup_{ \mathbf{m} \in B_{\rr^3}} \sup_{ \mathbf{m}' \in B_{\rr^3}} \bigg(\textrm{P}(\mathbf{m}) + g_P(\alpha) -\alpha \left( (\mathbf{m}_x - \mathbf{m}'_x)^2 + (\mathbf{m}_y - \mathbf{m}'_y)^2  +  (\mathbf{m}_z - \mathbf{m}'_z)^2 \right) - \Lambda_{\mathfrak{b}}^{*}( \pmb{m}') \bigg) \\
		&= \sup_{ \mathbf{m} \in B_{\rr^3}} \bigg(\textrm{P}(\mathbf{m}) - \Lambda_{\mathfrak{b}}^{*} (\pmb{m}) \bigg),
	\end{align*}
	which completes the proof.
	\end{proof}
	
	\subsection{Conclusion}\label{sec:con}
	
	So far, we have established the almost sure convergence
	$$ \lim_{N \to \infty} p_N(\textrm{P},\mathfrak{b}) = \sup_{\mathbf{m} \in B_{\rr^3}} \left(\textrm{P}(\mathbf{m}) - \Lambda^{*}_\mathfrak{b}(\mathbf{m}) \right)$$
	for all polynomial models. To complete the proof of Theorem~\ref{thm:main}, we need to extend this result to all continuous symbols $V$ and we further want to show that we also have $L^1$-convergence of the specific pressure. We start with the latter in the polynomial case.
	
	\begin{lemma}
			Let $\textrm{P} \colon \rr^3 \to \rr $ be a polynomial, $\mathfrak{b}$ an integrable vector valued random variable and $H_{\textrm{P},\mathfrak{b}}$ the associated random-field Hamiltonian(s) with specific pressures $p_N(\textrm{P},\mathfrak{b})$. Then, we have
		\begin{equation}\label{eq:L1}
			\lim_{N \to \infty} \mathbb{E} \left|  \sup_{\mathbf{m} \in B_{\rr^3}} \left(\textrm{P}(\mathbf{m}) - \Lambda^{*}_\mathfrak{b}(\mathbf{m}) \right)  	- p_N(\textrm{P},\mathfrak{b}) \right| = 0.
		\end{equation}
	\end{lemma}
  
    \begin{proof}
    	Combining Proposition~\ref{prop:lower} and Proposition~\ref{prop:upper}, we know that
    	$$ \lim_{N \to \infty} p_N(\textrm{P},\mathfrak{b})  = \sup_{\mathbf{m} \in B_{\rr^3}} \left(\textrm{P}(\mathbf{m}) - \Lambda^{*}_\mathfrak{b}(\mathbf{m}) \right)   \quad \text{a.s.} $$
    	and the only obstacle for the $L^1$-convergence is that we need to show that rare events do not contribute to the expectation of $p_N(\textrm{P},\mathfrak{b})$. We define the events
    	$$\Omega_N := \left\{ \omega \in \Omega \, \bigg| \, \sum_{n = 1}^{N} \mathbf{b}(n) > 2N \mathbb{E}[\mathfrak{b}]\right\}$$ for $N \in \nn$. We introduce the abbreviation $p_{\infty} := \sup_{\mathbf{m} \in B_{\rr^3}} \left(\textrm{P}(\mathbf{m}) - \Lambda^{*}_\mathfrak{b}(\mathbf{m}) \right) $ and employ the triangle inequality
    	\begin{align*}
    		\mathbb{E}[| p_N(\textrm{P},\mathfrak{b}) - p_{\infty} |] \leq \mathbb{E} [| \mathbbm{1}_{\Omega_N^c} (p_N(\textrm{P},\mathfrak{b}) - p_{\infty}) |] + \mathbb{E}[| \mathbbm{1}_{\Omega_N} p_{\infty}|] +
    		\mathbb{E}[| \mathbbm{1}_{\Omega_N} \, p_N(\textrm{P},\mathfrak{b})|] .
    	\end{align*}
    	Note that $\| H_P \| \leq CN$ and hence we obtain by an operator norm estimate 
    	\begin{equation}\label{eq:pbound} |p_N(\textrm{P},\mathfrak{b})| \leq \ln 2 + C + \frac{1}{N} \sum_{n = 1}^{N} \mathbf{b}(n).
    	\end{equation}
    	The strong law of large numbers implies that $\mathbbm{1}_{\Omega_N^c}$ converges almost surely to $1$ and employing further the bound \eqref{eq:pbound}, we have by dominated convergence $\lim_{N \to \infty} \mathbb{E} [| \mathbbm{1}_{\Omega_N^c} (p_N(\textrm{P},\mathfrak{b}) - p_{\infty}) |] =0 $.
    	The convergence $\lim_{N \to \infty} \mathbb{E}[| \mathbbm{1}_{\Omega_N} p_{\infty}|] = 0$ is immediate and it remains to estimate the last term.
    	To this end, we observe that 
    	$$ \mathbb{E}\left[ \mathbbm{1}_{\Omega_N} \frac1N \sum_{n = 1}^{N} |\mathbf{b}(n)| \right] = \mathbb{E}[|\mathfrak{b}|] - \mathbb{E}\left[ \mathbbm{1}_{\Omega_N^c} \frac1N \sum_{n = 1}^{N} |\mathbf{b}(n)| \right] $$
    	and by dominated convergence the right-hand side converges to zero. We apply the bound \eqref{eq:pbound}, which yields by the above $\lim_{N \to \infty} \mathbb{E}[| \mathbbm{1}_{\Omega_N} \, p_N(\textrm{P},\mathfrak{b})|] = 0.$ This completes the proof.
    	
    \end{proof}
    
    The final step towards the proof of Theorem~\ref{thm:main} is the following lemma.
    \begin{lemma}\label{lem:Vbound}
    	Let $H_V, H_{V'}$ be Hamiltonians meeting Assumption~\ref{ass:symbol} for two continuous symbols $V, \, V' : B_{\rr^3} \to \rr$. Then,
    	\begin{equation}\label{eq:Vbound}
    		\| H_V - H_{V'} \| \leq N \| V - V' \|_{\infty} + o(N),
    	\end{equation}
    	where $\| \cdot \|_{\infty}$ denotes the uniform norm on the space of functions.
    \end{lemma}
    \begin{proof}
    	Due to Assumption~\ref{ass:symbol}, it is enough to show 
    	$$ \left\|  \frac{2J+1}{4\pi} \int  N \left(V\Big(\frac{2J}{N}  \mathbf{e}(\Omega) - V'\Big(\frac{2J}{N}  \mathbf{e}(\Omega) \Big) \right) \,  \big| \Omega, J \rangle \langle \Omega, J \big| \, d\Omega  \right\| \leq N \| V - V' \|_{\infty} $$
    	for all $ J \in \{ \frac{N}{2} - \lfloor \frac{N}{2}\rfloor, \dots , N/2\} $. We recall that the above operator acts on $\cc^{2J +1}$ (or more precisely on a unitarily equivalent block). We pick a $\psi \in \cc^{2J +1}$ with $\| \psi \| = 1$ and estimate
    	\begin{align*}
    		&N \,\left| \frac{2J+1}{4\pi} \left\langle \psi, \left[\int   \left(V\Big(\frac{2J}{N}  \mathbf{e}(\Omega) - V'\Big(\frac{2J}{N}  \mathbf{e}(\Omega) \Big) \right) \,  \big| \Omega, J \rangle \langle \Omega, J \big| \, d\Omega \right]  \psi \right \rangle \right| \\
    		 = &N \, \left| \frac{2J+1}{4\pi}  \int   \left(V\Big(\frac{2J}{N}  \mathbf{e}(\Omega) - V'\Big(\frac{2J}{N}  \mathbf{e}(\Omega) \Big) \right) \,  |\langle \psi \big| \Omega, J \rangle |^2 \, d\Omega  \right| \\
    		 \leq &N \| V - V' \|_{\infty} \frac{2J+1}{4\pi} \, \int |\langle \psi \big| \Omega, J \rangle |^2 \, d\Omega  = N \| V - V' \|_{\infty},
    	\end{align*}
    	where we used the resolution of identity \eqref{eq:completeness} in the last step. Since this holds for any normalized $\psi$ the operator norm estimate follows.
    \end{proof}
    
    Completing the proof of Theorem~\ref{thm:main} is now straightforward.
    \begin{proof}[Proof of Theorem~\ref{thm:main}]
    	We use the notation $p(V,\mathfrak{b}) := \sup_{\mathbf{m} \in B_{\rr^3}} \left(V(\mathbf{m}) - \Lambda^{*}_\mathfrak{b}(\mathbf{m}) \right)$. We fix some $\varepsilon > 0$ and due to the Weierstrass theorem we may find a polynomial $\textrm{P}: B_{\rr^3} \to \rr$ such that $\|V - \textrm{P} \|_{\infty} < \varepsilon$. Then,
    	\begin{align*}
    		|p_N(V,\mathfrak{b}) - p(V,\mathfrak{b})| &\leq |p_N(P,\mathfrak{b}) - p(P,\mathfrak{b})| + |p_N(P,\mathfrak{b}) - p_N(V,\mathfrak{b})| + |p(V,\mathfrak{b}) - (P,\mathfrak{b})| \\
    		&\leq 2 \varepsilon + o(N) + |p_N(P,\mathfrak{b}) - p(P,\mathfrak{b})|,
    	\end{align*}
    	where we applied Lemma~\ref{lem:Vbound} and the simple bound $|p(V,\mathfrak{b}) - p(V',\mathfrak{b})| \leq \|V - V' \|_{\infty}$ for the limiting functional. Since $|p_N(P,\mathfrak{b}) - p(P,\mathfrak{b})|$ converges to zero almost surely and in expectation, the theorem follows as $\varepsilon >0$ can be chosen arbitrarily.
    \end{proof}

\appendix

\section{Supplementary Proofs}
In this appendix, we present some proofs which have been excluded from the main text. We start with Lemma~\ref{lem:iden} whose assertions are more or less folklore.
\begin{proof}[Proof of Lemma~\ref{lem:iden}]
	\begin{enumerate}
		\item Using the spectral decomposition, we have by definition
		\begin{align*}
			\langle \mathring{\mathbf{M}}^2_\mu \rangle_{H_{\textrm{P}, \mathfrak{b}}(\pmb{\gamma})} &= \tr \varrho_{{H_{\textrm{P}, \mathfrak{b}}(\pmb{\gamma})}} \mathring{\mathbf{M}}^2_\mu = \sum_{k} \langle \psi_k, \varrho_{{H_{\textrm{P}, \mathfrak{b}}(\pmb{\gamma})}} \mathring{\mathbf{M}}^2_\mu \psi_k \rangle \\
					&= \frac{1}{Z_N(H_{\textrm{P}, \mathfrak{b}}(\pmb{\gamma}))} \sum_{k} e^{E_k}  \langle \psi_{k}, \mathring{\mathbf{M}}^2_\mu \psi_k \rangle = \frac{1}{Z_N(H_{\textrm{P}, \mathfrak{b}}(\pmb{\gamma}))} \sum_{k,l} e^{E_k}  \langle \psi_{k}, \mathring{\mathbf{M}}_\mu \psi_l \rangle \langle \psi_{l}, \mathring{\mathbf{M}}_\mu \psi_k \rangle \\
					&= \frac{1}{Z_N(H_{\textrm{P}, \mathfrak{b}}(\pmb{\gamma}))} \sum_{k,l} e^{E_k} |M^{\mu}_{k,l}|^2 = \frac{1}{2 Z_N(H_{\textrm{P}, \mathfrak{b}}(\pmb{\gamma}))} \sum_{k,l} (e^{E_k} + e^{E_l}) |M^{\mu}_{k,l}|^2,
		\end{align*}
		where we used the resolution of identity $\mathbbm{1} = \sum_{l} | \psi_l \rangle \langle \psi_l |$ in the second line and the definition of $|M^{\mu}_{k,l}|$ and symmetry in the third line.
		\item We employ again the spectral decomposition
		\begin{align*}
			\langle [\mathbf{M}_\mu, H_{\textrm{P}, \mathfrak{b}}(\pmb{\gamma})]^2 \rangle_{H_{\textrm{P}, \mathfrak{b}}(\pmb{\gamma})} &= \frac{1}{Z_N(H_{\textrm{P}, \mathfrak{b}}(\pmb{\gamma}))} \sum_{k,l} e^{E_k}  \langle \psi_{k}, [\mathbf{M}_\mu, H_{\textrm{P}, \mathfrak{b}}(\pmb{\gamma})]  \psi_l \rangle \langle   \psi_{l},  [\mathbf{M}_\mu, H_{\textrm{P}, \mathfrak{b}}(\pmb{\gamma})]\psi_k \rangle \\
			&= \frac{1}{Z_N(H_{\textrm{P}, \mathfrak{b}}(\pmb{\gamma}))} \sum_{k,l} e^{E_k} M^{\mu}_{k,l} (E_l - E_k) M^{\mu}_{l,k} (E_k - E_l)
			\\ &= - \frac{1}{2 Z_N(H_{\textrm{P}  \mathfrak{b}}(\pmb{\gamma}))} \sum_{k,l = 1}^{2^N} |M_{k,l}^{\mu}|^2 (e^{E_k} + e^{E_l})|E_k - E_l|^2,
		\end{align*}
		where the first line follows as the in the previous part, the second line is immediate and the third line is again by symmetry.
		\item Let us first give a coordinate-free description of $\frac{\partial^2}{\partial \gamma_u^2} p_N(H_{\textrm{P}, \mathfrak{b}}(\pmb{\gamma}))$. Clearly $\frac{\partial^2}{\partial \gamma_u^2} p_N = \frac{ \partial_{\gamma_u}^2 Z_N}{Z_N} - \left(\frac{ \partial_{\gamma_u} Z_N}{Z_N} \right)^2$
		and the first derivative is well known to be $ \partial_{\gamma_u} Z_N = \tr \mathbf{M}_{\mu} e^{H}$. For the second derivative, we invoke the Duhamel formula to obtain
		$$ \partial_{\gamma_u}^2 Z_N(H_{\textrm{P}, \mathfrak{b}}(\pmb{\gamma})) = \int_{0}^{1} \tr \left[ e^{(1-s)H_{\textrm{P}, \mathfrak{b}}(\pmb{\gamma})} \mathbf{M}_\mu e^{s H_{\textrm{P}, \mathfrak{b}}(\pmb{\gamma})} \mathbf{M}_\mu \right] \, ds. $$
		Let us denote by $\widehat{M}^{\mu}_{k,l}$ the matrix elements of $\mathbf{M}_\mu$ in the eigenbasis $\psi_k$ of $H_{\textrm{P}, \mathfrak{b}}(\pmb{\gamma})$. Note that the coefficients $M^{\mu}_{k,l}$ refer in contrast to the centered observable $ \mathbf{M}_\mu - \langle \mathbf{M}_\mu \rangle_{H_{\textrm{P}, \mathfrak{b}}(\pmb{\gamma})}$. As in the first part, we arrive at 
		\begin{align*} \partial_{\gamma_u}^2 Z_N(H_{\textrm{P}, \mathfrak{b}}(\pmb{\gamma})) &= \sum_{k,l} \int_{0}^{1} e^{(1-s)E_k + s E_l} |\widehat{M}^{\mu}_{k,l}|^2 \, ds = \sum_{k} |\widehat{M}^{\mu}_{k,k}|^2 e^{E_k} + \sum_{k \neq l} |\widehat{M}^{\mu}_{k,l}|^2 \frac{e^{E_k} - e^{E_l}}{E_k - E_l} \\ 
			&= \sum_{k} |\widehat{M}^{\mu}_{k,k}|^2 e^{E_k} + \sum_{k \neq l} |M^{\mu}_{k,l}|^2 \frac{e^{E_k} - e^{E_l}}{E_k - E_l}
		\end{align*}
		On the other hand, $\left(\frac{ \partial_{\gamma_u} Z_N}{Z_N} \right)^2 = \left(\langle \mathbf{M}_\mu \rangle_{H_{\textrm{P}, \mathfrak{b}}(\pmb{\gamma})} \right)^2$ and  $$\frac{1}{Z_N} \sum_{k} |\widehat{M}^{\mu}_{k,k}|^2 e^{E_k} = \frac{1}{Z_N} \sum_{k} |M^{\mu}_{k,k}|^2 e^{E_k} + \left(\langle \mathbf{M}_\mu \rangle_{H_{\textrm{P}, \mathfrak{b}}(\pmb{\gamma})} \right)^2$$. 
		The claim follows by combining the above identities.
		
	\end{enumerate}
\end{proof}

The Lemma~\ref{lem:cosh} contains an elementary inequality which is proved next.

\begin{proof}[Proof of Lemma~\ref{lem:cosh}]
	We observe that for $x > 0$
	\begin{align*}
		\cosh x &= \sinh x + e^{-x} = \sinh x + \frac{\sinh x}{x} + \frac{x e^{-x} - \sinh(x)}{x} \\
		&= \sinh x + e^{-x} = \sinh x + \frac{\sinh x}{x} + \frac{e^{-x}\left(x+\frac12 -e^{2x}\right)}{x} 
	\end{align*}
   and the last term is negative due to the elementary inequality $e^{u} \geq 1 + u$. Hence, we arrive at $\cosh x \leq \sinh x + \frac{\sinh x}{x}.$
\end{proof}

The proof of Lemma~\ref{lem:conv} relies on the following basic fact.
\begin{lemma}\label{lem:convv}
	Let $K \subset \rr^d$ be a compact set and $(\Omega, \mathcal{F}, \mathbb{P})$ a fixed probability space. Let further $X_N : K \times \Omega \to \rr $ be a sequence of real-valued maps with the following properties.
	\begin{enumerate}
		\item For each fixed $z \in K$, the maps $\omega \mapsto X_N(z, \omega)$ are random variables which converge almost surely to a random variable $X(z,\omega)$.
		\item The maps $X_N$ are uniformly Lipschitz in the first variable, i.e.,  there is an $L > 0$ independent of $z, \omega,N$ such that 
		\begin{equation}
			|X_N(z,\omega) - X_N(z', \omega)| \leq L |z -z'|
		\end{equation}
		for all $z,z' \in K, \omega \in \Omega, N \in \nn.$
	\end{enumerate}
	Then, the random variables $\inf_{z \in K} X_N(z,\omega)$ converge almost surely to $ \inf_{z \in K} X(z,\omega)$.
\end{lemma} 
\begin{proof}
	Due to continuity, $\inf_{z \in K} X_N(z,\omega)$  and $ \inf_{z \in K} X(z,\omega)$ are measurable for all $N$. Thanks to compactness, for any fixed $\varepsilon >0$ there exists a finite collection of points $z_1, \ldots, z_M \in K$ such that $ \min_{ 1 \leq m \leq  M} |z - z_m | < \varepsilon$ for every $z \in K$. We employ the Lipschitz bound to obtain
	\begin{align*}
		\left| \inf_{z \in K} X_N(z,\omega) - \inf_{z \in K} X(z,\omega) \right| \leq 2 L \varepsilon + \left| \min_{1 \leq m \leq M} X_N(z_m,\omega) - \min_{1 \leq m \leq M} X(z_m,\omega) \right| 
	\end{align*}
	and the last term converges almost surely to zero as $N \to \infty$. Since $\varepsilon > 0$ is arbitrary, the claim follows.
\end{proof}

We are now ready to spell out the proof of Lemma~\ref{lem:conv}.
\begin{proof}[Proof of Lemma~\ref{lem:conv}]
	In view of \eqref{eq:pconv} and Lemma~\ref{lem:convv}, it is enough to verify the Lipschitz continuity for the linear pressure $ p_N(L(\pmb{\alpha}, \pmb{m}, \{\mathbf{b}(n)\}_{n = 1}^{N}))$. We note that
	$$ L(\pmb{\alpha}, \pmb{m}, \{\mathbf{b}(n)\}_{n = 1}^{N}) - L(\pmb{\alpha}, \pmb{m}', \{\mathbf{b}(n)\}_{n = 1}^{N}) = 4 \left(\alpha_x (m_x'-m_x) \mathbf{S}_x + \alpha_y (m_y'-m_y) \mathbf{S}_y + \alpha_z (m_z'-m_z) \mathbf{S}_z \right)  $$ 
	for $\mathbf{m}, \mathbf{m}' \in B_{\rr^3}$. In particular,
	\begin{align*}
		|p_N(L(\pmb{\alpha}, \pmb{m}, \{\mathbf{b}(n)\}_{n = 1}^{N})) - p_N(L(\pmb{\alpha}, \pmb{m}', \{\mathbf{b}(n)\}_{n = 1}^{N}))| &\leq \| L(\pmb{\alpha}, \pmb{m}, \{\mathbf{b}(n)\}_{n = 1}^{N}) - L(\pmb{\alpha}, \pmb{m}', \{\mathbf{b}(n)\}_{n = 1}^{N})   \| \\ & \leq 2 \left( \max_{\mu = x,y,z} \alpha_\mu \right) \| \pmb{m} - \pmb{m}' \| 
	\end{align*}
	which establishes the uniform Lipschitz continuity.
\end{proof}
\minisec{Acknowledgments}
This work was funded by the Deutsche Forschungsgemeinschaft (DFG, German Research Foundation) - 558731723.

\bibliographystyle{plain}

\end{document}